\documentclass[a4paper, 12pt]{article}

\usepackage{fullpage}
\usepackage{amsthm}
\theoremstyle{plain}
\newtheorem{theorem}{Theorem}
\newtheorem{lemma}[theorem]{Lemma}
\newtheorem{corollary}			[theorem]	{Corollary}

\newtheorem{notation}			[theorem]	{Notation}
\newtheorem{fact}				[theorem]	{Fact}

\theoremstyle{definition}
\newtheorem{definition}[theorem]{Definition}
\theoremstyle{definition}
\newtheorem{remark}[theorem]{Remark}

\newcommand{\squarerel}{\mathbin{\square}}
\newcommand{\symdiff}{\mathbin{\Delta}}
\newcommand{\dunion}{\mathbin{\dot\cup}}

	\usepackage[utf8]{inputenc}
	\usepackage[T1]{fontenc}
	\usepackage[ngerman,american]{babel}
	\usepackage[usenames,dvipsnames,svgnames,table]{xcolor}
	\usepackage{amsmath,enumitem,leftidx,framed,graphicx,hyperref,clrscode,parskip,changepage,tikz,float}
	\usepackage{hyperref}
	\usepackage{graphicx}
		\graphicspath{ {./pictures/} }
		\hypersetup{colorlinks, linkcolor=teal, citecolor=teal, urlcolor=teal}
	\usepackage{amssymb}

		\makeatletter
			\def\thm@space@setup{%
			  \thm@preskip=\parskip \thm@postskip=0pt
			}
		\makeatother
\newcommand	{\abs}			[1]		{\left\lvert#1\right\rvert}
\newcommand{\dfn}[1]		{\textbf{#1}}
\renewenvironment{proof}{{\noindent\bfseries Proof:}}{\qed}

\usepackage{xspace}

\usepackage{mathtools}
\usepackage{subfig}

\usepackage{etex,etoolbox}

\definecolor{darkred}{rgb}{0.5,0,0}

\title{Graphs Identified by Logics with Counting}

\author{Sandra  Kiefer, Pascal Schweitzer and Erkal Selman \\ RWTH Aachen University\\
{\tt \{kiefer,schweitzer,selman\}@informatik.rwth-aachen.de} \\
  Ahornstra\ss{}e 55, 52074 Aachen, Germany}

\DeclareMathOperator{\even}{even}
\DeclareMathOperator{\Circ}{Circ}
\DeclareMathOperator{\Match}{Match}
\DeclareMathOperator{\K33}{\equiv_3^3}

\newcommand{\LogicCK}{C^k}
\newcommand{\LogicCzwo}{C^2}
\newcommand{\LogicCdrei}{C^3}
\newcommand{\LogicCKminusFour}{C^{k-4}}

\DeclareMathOperator{\lcm}{lcm}
\DeclareMathOperator{\ecPOG}{ec-POG}

\begin{document}
\maketitle

	\begin{abstract}
We classify graphs and, more generally, finite relational structures that are identified by~$\LogicCzwo$, that is, two-variable first-order logic with counting. 
Using this classification, we show that it can be decided in almost linear time whether a structure is identified by $\LogicCzwo$.
Our classification implies that for every graph identified by this logic, all vertex-colored versions of it are also identified. A similar statement is true for finite relational structures.

We provide constructions that solve the inversion problem for finite structures in linear time. This problem has previously been shown to be polynomial time solvable by Martin Otto. 
For graphs, we conclude that every~$\LogicCzwo$-equivalence class contains a graph whose orbits are exactly the classes of the~$\LogicCzwo$-partition of its vertex set and which has a single automorphism witnessing this fact.

For general $k$, we show that such statements are not true by providing examples of graphs of size linear in~$k$ which are identified by~$\LogicCdrei$ but for which the orbit partition is strictly finer than the~$\LogicCK$-partition. We also provide identified graphs which have vertex-colored versions that are not identified by~$\LogicCK$.
	\end{abstract}


\section{Introduction}

The $k$-variable fragment of counting logic, denoted by~$\LogicCK$, is obtained from first-order logic by
adding counting quantifiers 
but only allowing formulas that use at most~$k$ variables.
These finite variable logics play a central role in the area of model-checking since 
for them the model-checking problem and the  equivalence problem are solvable in polynomial time (see~\cite{DBLP:journals/bsl/Grohe98}). 
For a while, there was the hope that for some fixed~$k$ the logic~$\LogicCK$ can distinguish every pair of non-isomorphic graphs. This would imply that the graph isomorphism problem is solvable in polynomial time.
However, in 1992, it was shown by Cai, Fürer and Immerman~\cite{DBLP:journals/combinatorica/CaiFI92} 
that $\Omega(n)$ variables are required to identify all graphs on $n$ vertices. Since the examples presented in that paper consist of graph isomorphism instances which are actually known to be solvable in polynomial time, this also shows that~$\LogicCK$ does not capture polynomial time.

Concerning~$\LogicCK$, there are striking connections to other seemingly unrelated areas.
For example, there exist several Ehrenfeucht-Fra\"{i}ss\'e type games characterizing $\LogicCK$ \cite{DBLP:journals/combinatorica/CaiFI92,DBLP:conf/icalp/DawarH12,DBLP:journals/iandc/Hella96}. 
Also strongly related is the $(k-1)$-dimensional version of a well-known color refinement algorithm,
named after Weisfeiler and Lehman by Babai (see~\cite{DBLP:journals/combinatorica/CaiFI92}).
It turns out that the $(k-1)$-dimensional Weisfeiler-Lehman algorithm
does nothing else but partition $(k-1)$-tuples of vertices
according to their $\LogicCK$-types.
Another surprising connection exists to linear programming. The~$k$-th level of the Sherali-Adams hierarchy of a natural linear integer programming formulation of graph isomorphism essentially corresponds to the expressive power of~$\LogicCK$ \cite{DBLP:journals/siamcomp/AtseriasM13,DBLP:conf/csl/GroheO12}.

Among the finite variable logics, the fragment~$\LogicCzwo$ has been of particular interest because it is known to be decidable~\cite{DBLP:conf/lics/GradelOR97}, and the complexity of the decision problem has been studied extensively~\cite{DBLP:journals/jolli/PrattHartmann05}. Numerous results for this logic are known, we refer the reader to a survey by Gr{\"{a}}del and Otto~\cite{DBLP:journals/tcs/GradelO99}.
In practice, due to its strength and the fact that it can be evaluated in almost linear time, the logic~$\LogicCzwo$ (more specifically, the corresponding 1-dimensional Weisfeiler-Lehman algorithm) is an essential subroutine in all competitive canonical labeling tools~(see~\cite{DBLP:journals/jsc/McKayP14}).
Very recent results concerning~$\LogicCzwo$ include a paper by Krebs and Verbitsky studying the quantifier depth of $\LogicCzwo$-formulas for $\LogicCzwo$-equivalence classes of graphs~\cite{DBLP:journals/corr/KrebsV14}. Kopczynski and Tan show that for every fixed $\LogicCzwo$-formula, the set of those~$n$ for which there is a structure with a universe of size~$n$ satisfying the formula is semilinear~\cite{DBLP:journals/corr/abs-1304-0829}.

While the above results deal with the problem of distinguishing two graphs from each other using finite variable counting logics, in this paper we are concerned with the concept of distinguishing a graph from every other non-isomorphic graph. We say that the graph is \dfn{identified} by the logic. More formally, a graph (or a finite relational structure) $G$ is identified by a logic $L$ if there is a sentence $\varphi$ in $L$ such that every graph (or finite relational structure) which satisfies $\varphi$ is isomorphic to $G$.

Of course every graph is identified by some first-order sentence. However, by \cite{DBLP:journals/combinatorica/CaiFI92}, as mentioned above, there is no~$k \in \mathbb{N}$ such that every graph is identified by some formula in~$\LogicCK$. Let us focus on the case~$k=2$. It is not difficult to see that all trees are identified by~$\LogicCzwo$. Moreover, a graph is asymptotically almost surely identified by~$\LogicCzwo$, that is, the fraction of graphs of size~$n$ which are not identified by~$\LogicCzwo$ tends to~$0$ as~$n$ tends to infinity~\cite{DBLP:journals/siamcomp/BabaiES80}. Even more strongly, it is known~\cite{DBLP:conf/focs/BabaiK79} that the fraction of graphs which are not identified is exponentially small in~$n$. Similarly, a regular graph is asymptotically almost surely identified by~$\LogicCdrei$~\cite{DBLP:conf/focs/Kucera87}. 
	However, not all graphs are identified. (For~$\LogicCzwo$, consider a cycle of length at least~$6$, for example.) The following question arises.

 \emph{What is the structure of graphs that are identified by~$\LogicCK$?}

\paragraph{Our results.} 
We study graphs that are identified by~$\LogicCzwo$ 
	and provide a complete classification for them. 
This classification can be used to draw several conclusions 
	about general properties of identified graphs.
For example, one can derive that if an undirected graph is identified by~$\LogicCzwo$, 
	then the~$\LogicCzwo$-partition classes of the vertices 
	are exactly the orbits of the automorphism group of the graph. 
This corollary is neither true when considering finite (relational) structures 
	nor when considering~$\LogicCK$ with~$k>2$. 
For~$\LogicCzwo$, we also conclude that 
	if an undirected graph is identified, every vertex-colored version of it is identified by~$\LogicCzwo$ as well.
This statement holds for finite relational structures, too, 
	but is again not true for~$\LogicCK$ with~$k>2$. 
Using our classification, we show that in time~$O((n+m)\log n)$ 
	it is possible to determine 
	whether an undirected graph is identified by~$\LogicCzwo$.

To prove the correctness of our classification, we need explicit constructions that solve the inversion problem and the canonization problem for~$\LogicCzwo$. The inversion problem asks whether to a certain invariant a graph (or more generally, a finite structure) can be constructed. In the case of~$\LogicCzwo$, such an invariant is the count of the~$\LogicCzwo$-types of pairs of vertices~$x,y$. A celebrated result by Otto~\cite{DBLP:journals/apal/Otto97} shows that the inversion problem and the canonization problem for~$\LogicCzwo$ can be solved in polynomial time. 
As a by-product, our direct constructions provide an alternative proof for this. In fact, they show that the inversion problem for~$\LogicCzwo$ can be solved in linear time. Our constructions make use of circulant graphs and doubly-circulant graphs. With these, we observe that every~$\LogicCzwo$-equivalence class contains a graph whose $\LogicCzwo$-partition classes are the orbits. More strongly, there is a single automorphism of the graph witnessing this. To achieve inversion for finite structures, we use an old 1-factorization construction due to Walecki~(see \cite{MR0124983}) that decomposes the complete graph~$K_{2n}$ into~$2n-1$ disjoint perfect matchings.

Building on the classification of graphs identified by~$\LogicCzwo$, we also classify finite structures that are identified by~$\LogicCzwo$. For graphs, there is only one special case that may appear within a $\LogicCzwo$-partition class (namely the cycle of length 5). However, for finite structures there are $7$ different special cases, which are of sizes~$3$,~$4$,~$5$ and~$6$ for a~$\LogicCzwo$-partition class. Our classification theorem describes how these may be combined to form structures that are identified by~$\LogicCzwo$. Due to the nature of the different special cases, the classification is more involved (see Theorem~\ref{thm:classification:of:finite:ident:sturctures}). Nevertheless, we can show that one can decide in almost linear time whether a structure is identified by~$\LogicCzwo$.

For the logics~$\LogicCK$ with~$k>2$ we collect several negative results. 
One can first observe that the triangular graphs form an infinite non-trivial class of strongly regular graphs which are identified by~$\LogicCdrei$, implying that any classification result would have to include non-trivial infinite families~\cite{chang,hoffman1960}.

Contrasting our results for~$\LogicCzwo$, we provide examples of graphs that are identified by~$\LogicCdrei$ but for which conclusions analogous to the ones mentioned above do not hold.  More specifically, we present graphs identified by~$\LogicCdrei$ for which even the logic~$\LogicCK$ with~$k$ linear in the size of the graph does not correctly determine the orbit partition. This yields graphs which the logic~$\LogicCK$ identifies, but for which not all vertex-colored versions are identified by~$\LogicCK$. These ideas are based on the construction by Cai, F{\"{u}}rer and Immerman~\cite{DBLP:journals/combinatorica/CaiFI92}. 

The existence of these graphs highlights an important fact. Even if a graph is identified by the logic~$\LogicCK$, it is not clear that it is possible to take advantage of that in order to canonize the graph. This stands in contrast to a remark in~\cite{DBLP:journals/combinatorica/CaiFI92} claiming that a graph~$G$ identified by~$\LogicCK$ can be canonized in polynomial time. In fact, the crucial property required for the approach hinted at there to be successful is that all vertex-colored versions of~$G$ need to be identified by~$\LogicCK$. Indeed, if this property holds then a standard recursive individualization approach canonizes the graph~$G$. It would suffice to show that for all vertex-colored versions of~$G$ the orbits are determined by the logic. However, with a slight alteration of our construction, we obtain a graph~$G$ that is identified by~$\LogicCK$ and whose orbits are correctly determined, but for some vertex-colored versions of~$G$ the orbits are not correctly determined.

Independently of our work, Arvind, K\"{o}bler, Rattan and  Verbitsky~\cite{2015arXiv150201255A} have investigated the structure of undirected graphs identified by~$\LogicCzwo$ obtaining results similar to the ones we provide in Section~\ref{sec:characterization:of:graphs}.

\paragraph{Organization of the paper.} After providing preliminaries in Section~\ref{sec:prelims}, we give constructions for graphs and finite relational structures with given color degrees in Section~\ref{sec_inversion}, also solving the inversion problem. We then classify graphs identified by~$\LogicCzwo$ in Section~\ref{sec:characterization:of:graphs} and use these results to classify identified finite relational structures in Section~\ref{sec:general:finite:structures}. Finally, in Section~\ref{sect_highdim} we collect negative results on logics~$\LogicCK$ for~$k>2$.


\section{Preliminaries}\label{sec:prelims}

Unless specified otherwise, by a \dfn{graph} we always mean a finite undirected graph without loops.
We denote the vertex set of a graph $G$ by $V(G)$
	and the edge set by $E(G)$.
For a subset $P$ of $V(G)$, we denote by \dfn{$G[P]$}
	the subgraph of $G$ induced by $P$.
The number of neighbors of a vertex is the \dfn{degree} of the vertex.
If all vertices have degree~$k$, we call the graph \dfn{$k$-regular}.
A \dfn{$(k,\ell)$-biregular graph $G$ on bipartition $(P,Q)$}
	is a graph on vertex set $P\dunion Q$ such that
	$P$ and $Q$ are independent sets,
	every vertex in $P$ has exactly $k$ neighbors in $Q$ and
	every vertex in $Q$ has exactly $\ell$ neighbors in $P$.

Two vertices $v, v'\in V(G)$ are in the same \dfn{orbit of $G$} (more precisely, in the same orbit of the automorphism group of~$G$)
	if $G$ has an automorphism $\varphi$ such that $\varphi(v) = v'$. 
The \dfn{orbit partition} of $G$ is the partition of $V(G)$ into the orbits of $G$.
	
A graph $G$ is \dfn{identified} by a logic $L$
	if there is a sentence $\varphi$ in~$L$ 
	such that every graph which satisfies $\varphi$ 
	is isomorphic to $G$.
A logic $L$ \dfn{distinguishes} two graphs $G$ and $G'$
	if there is a sentence $\varphi$ in $L$ such that $G\models \varphi$ and $G'\not\models \varphi$.
We say that $G$ and $G'$ are \dfn{$L$-equivalent} if they are not distinguished by $L$. 

The \dfn{$k$-variable counting logic}, denoted by $\LogicCK$,
	is the $k$-variable fragment of first-order logic enriched by counting quantifiers.
For every $t\in \mathbb{N}$ we have the counting quantifier~$\exists^{\geq t}$.
For a formula $\varphi(x)$ with the free variable $x$ and for a graph $G$
	we have $G \models \exists^{\geq t}x\ \varphi(x)$ 
	if and only if
	there are at least $t$ vertices $v\in V(G)$ such that $G\models \varphi[v]$ (where, as usual,~$\varphi[v]$ denotes the formula obtained by substituting~$v$ for the free variable~$x$ in~$\varphi$).
The (bound and free) variables in a $\LogicCK$-formula are all from a fixed $k$-element set,
	say $\{x_1, \ldots, x_k\}$, but they can be reused.
For example, the formula 
	\[
		\exists x_1 \exists^{\geq 2} x_2 (E(x_1,x_2) \land \exists^{\geq 5} x_1 E(x_1,x_2))
	\]
	is a valid $\LogicCzwo$-sentence which says that
	there is a vertex with at least 2 neighbors of degree at least~$5$.
The \dfn{$\LogicCK$-type} of a vertex $v$ in $G$
	is the set of all $\LogicCK$-formulas $\varphi(x)$ 
	such that $G\models \varphi[v]$.
The \dfn{$\LogicCK$-coloring} of $G$ 
	is the coloring of each vertex with its $\LogicCK$-type.
The \dfn{$\LogicCK$-partition} of a graph $G$ 
	is the partition of its vertex set induced by their $\LogicCK$-types.
Similarly, the \dfn{$\LogicCK$-type} of a tuple $(v, w)$ 
	is the set of $\LogicCK$-formulas $\varphi(x,y)$
	such that $G\models \varphi[v,w]$. For a vertex~$v$ or a pair of vertices~$v$ and~$w$ to obtain the \dfn{atomic $\LogicCK$-type} we consider only quantifier-free~$\LogicCK$-formulas.

Let $G$ be a graph and $\Pi$ be a partition of $V(G)$.
We say that $\Pi$ is \dfn{equitable}
	if for all $P,Q\in \Pi$ and $v,v'\in P$, 
	the vertices $v$ and $v'$ have the same number of neighbors in $Q$.
A vertex coloring $\chi$ is called equitable 
	if the partition induced by $\chi$ is equitable.

It is a well-known fact 
	that the~$\LogicCzwo$-partition of a graph 
	is its~\dfn{coarsest equitable partition}.
The~$\LogicCzwo$-partition of a graph with $n$ vertices and $m$ edges can be calculated in time $O((m+n)\log n)$ 
	by the color refinement procedure (see \cite{DBLP:conf/esa/BerkholzBG13}), 
	also called na\"{\i}ve vertex classification or 1-dimensional Weisfeiler-Lehman algorithm.

\subsection{Relational structures and partially oriented graphs}\label{subsec:rel:struct:and:pocs:prelims}

In order to simplify the task of 
	generalizing our results to finite relational structures,
	we introduce edge-colored partially oriented graphs.
For every structure $\mathfrak{A}$ 
	we define an edge-colored partially oriented graph $\ecPOG(\mathfrak{A})$
	that holds all the information about $\mathfrak{A}$ 
	expressible by using only two variables.

An \dfn{edge-colored partially oriented graph} (in short, an \dfn{$\ecPOG$})
	is an edge-colored directed graph $(G,c)$ (with~$c$ an edge-coloring function) without loops
	such that for every $(v,w)\in E(G)$, it holds that 
	if $(w,v)\in E(G)$ then~$c((v,w)) = c((w,v))$.
Slightly abusing terminology,
	we say that an edge $(v,w)\in E(G)$ is \dfn{undirected} if $(w,v)\in E(G)$
	and \dfn{directed} otherwise. We accordingly draw~$(v,w)$ and~$(w,v)$ as one undirected edge between~$v$ and~$w$ (see Figure~\ref{fig:POC}) and denote it by~$\{v,w\}$.
An $\ecPOG$ $(G,c)$ is \dfn{complete} 
	if for all $v$, $w \in V(G)$ with $v \neq w$ we have ${(v,w)\in E(G)}$ or ${(w,v)\in E(G)}$.

\begin{figure}[H] 
	\centering \includegraphics[width=0.2\textwidth]{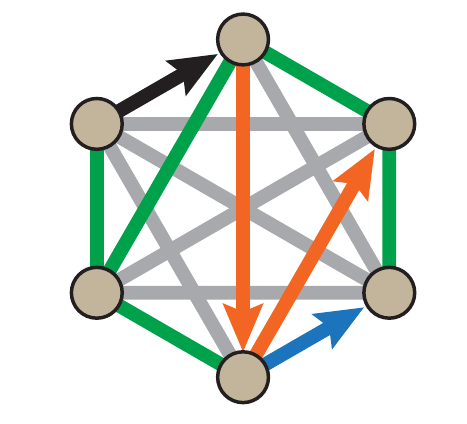}
	\caption{a complete $\ecPOG$}\label{fig:POC}
\end{figure}	

In the following, we consider finite relational structures
	over a fixed signature $\sigma = (R_1, \ldots, R_\ell)$ 
	where~$R_i$ has arity~$r_i$. The various definitions given for graphs are  analogously defined for structures. Let~$\mathfrak{A}$ be a finite relational structure with universe~$A$.
For every $i\in \{1, \ldots, \ell\}$ we define 
a function~${c_i\colon A^2 \to \mathcal{P}(\{1,2\}^{r_i})}$ via
	\[
		c_i (v_1, v_2) \coloneqq
			\left\{ (j_1, \ldots, j_{r_i})\in \{1,2\}^{r_i}
			\mathrel{}\middle|\mathrel{}\mathfrak{A} \models R_i(v_{j_1}, \ldots, v_{j_{r_i}}) 
			\right\}
	\]
	where, as usual,
	~$\mathcal{P}$ denotes the power set
	and~$\{1,2\}^{r_i}$ 
	denotes the set of all $r_i$-tuples over~$\{1,2\}$.
For every $v$, $w \in A$ with $v\neq w$ we let
\[
	c(v,w)\coloneqq (c_1(v,w), \ldots, c_\ell(v,w)).
\]	
Since for each~$i$ the possible images of~$c_i$ come from a set of bounded size, by using the order of the relations~$R_i$ in~$\sigma$,
	one can easily define a canonical linear ordering~$\leq$ on the image of~$c$ (for example, by using the lexicographic order). 
With the help of this ordering, 
	we define ${\ecPOG(\mathfrak{A}) \coloneqq ((A, E_\mathfrak{A}), c_\mathfrak{A})}$ 
	as the complete $\ecPOG$ 
	with vertex set $A$,
 edge set
	\[
		E_\mathfrak{A} \coloneqq \{ (v,w) \mid v,w \in A, v\neq w \text{ and } c(v,w)\leq c(w,v) \}
	\]
	and the edge coloring $c_\mathfrak{A} \coloneqq \left.c\right|_{E_\mathfrak{A}}$, 
	the restriction of $c$ to the domain $E_\mathfrak{A}$.

Note that~$c(v,w)$ uniquely determines 
	the atomic $\LogicCzwo$-types of $v$, $w$, $(v,w)$ and $(w,v)$.
Intuitively, this means that if the universe $A$ has at least two elements,
	then the coloring $c$ stores all the information about $\mathfrak{A}$
	that the logic~$\LogicCzwo$ can express.
Hence, there is no need to include vertex colors or loops in our definition of $\ecPOG$s. (More concretely, the atomic $\LogicCzwo$-type of~$v$ is encoded in the color~$c(v,w)$ for all~$w\neq v$, so does not need to be modeled by loops or vertex colors.)
Note also that in $\ecPOG(\mathfrak{A})$,
	a directed edge cannot have the same color as an undirected edge.

A partition $\Pi$ of the vertex set of an $\ecPOG$ is \dfn{equitable}
	if for all $P, Q\in \Pi$, for all $v,v'\in P$ and for every edge color $c$,
	the vertices $v$ and $v'$ have the same number of $c$-colored outgoing, incoming and undirected edges connecting them to~$Q$.
An $\ecPOG$ is \dfn{color-regular}
	if for each edge color $c$ every vertex $v$ has 
		the same $c$-indegree,
		the same $c$-outdegree and
		the same $c$-degree for undirected edges, respectively.
An edge-colored undirected biregular graph on bipartition~$(P,Q)$
	is called \dfn{color-biregular} if 
	for every edge color $c$ 
	the subgraph induced by the edges of color $c$ is biregular on~$(P,Q)$. 
If the graph is partially oriented, vertices in each bipartition class must additionally have the 		same number of outgoing and incoming edges in each color.
Note that for an equitable partition of an $\ecPOG$,
	the graph induced by one~$\LogicCzwo$-partition class is always color-regular 
	and the graph induced between two~$\LogicCzwo$-partition classes is color-biregular.


\section{Inversion}\label{sec_inversion}

In this section, 
	we treat the so-called \emph{inversion problem} that is closely related to the question which graphs are identified by~$\LogicCzwo$.
A \dfn{complete invariant} of an equivalence relation~$\equiv$
	on a class~$\mathcal{C}$ of structures 
	is a mapping~$\mathcal{I}$ from~$\mathcal{C}$ to some set~$S$,
	such that~$\mathfrak{A}\equiv \mathfrak{B}$ 
	if and only if~$\mathcal{I}(\mathfrak{A}) = \mathcal{I}(\mathfrak{B})$.
We say that \dfn{$\mathcal{I}$ admits linear time inversion}
	if given~$s\in S$ one can construct in linear time 
	a structure~$\mathfrak{A}$ with~$\mathcal{I}(\mathfrak{A}) = s$ or decide that no such structure exists. This algorithmic task describes the inversion problem.
	Of course, if a structure~$\mathfrak{A}$ is identified then a solution to the inversion problem must construct~$\mathfrak{A}$ when given~$\mathcal{I}(\mathfrak{A})$.

For~$\LogicCzwo$, we show that a natural complete invariant, namely~$\mathcal{I}_C^2$, admits linear time inversion.
Otto~\cite{DBLP:journals/apal/Otto97} proved
	that this invariant admits polynomial time inversion.

Given a graph~$G$, 
	one can define a linear ordering 
	$P_1 \leq \ldots \leq P_t$ on the classes of its coarsest equitable partition,
	which only depends on the~$\LogicCzwo$-equivalence class of~$G$
	(see~\cite{DBLP:journals/apal/Otto97}).
This ordering allows us to define~$\mathcal{I}_C^2$, 
	mapping~$G$ to~$(\bar{s},M)$, 
	where~$\bar{s}$ is the tuple~$(\abs{P_1}, \ldots, \abs{P_t})$ and
	$M$ is a~$t \times t$ matrix,
	such that every vertex in~$P_i$
	has exactly~$M_{ij}$ neighbors in~$P_j$.
It is easy to see that~$\mathcal{I}_C^2$ is a complete invariant of~$\LogicCzwo$.

\subsection{Inversion for graphs}\label{subsec:grap:inv}

\begin{definition}
	A graph is \dfn{circulant} 
		if it has an automorphism with exactly one cycle. 
	A graph on vertex set~$P \dunion Q$
		is \dfn{doubly-circulant} with respect to~$P$ and~$Q$
		if it has an automorphism with exactly two cycles, 
		one on~$P$ and the other on~$Q$.
	Analogously, a graph is \dfn{multi-circulant} 
		with respect to a partition~$\{P_1,\ldots, P_\ell\}$ of its vertex set
		if it has an automorphism with exactly~$\ell$ cycles, 
		each on one of the~$P_i$.
\end{definition}

It is well-known that circulant graphs can be constructed 
	by numbering the vertices from~$0$ to~$n-1$, 
	picking a set~$S \subseteq \{1, \dots, \lfloor n/2\rfloor\}$ of distances
	and inserting all edges between pairs of vertices 
	whose distance of indices in the circular ordering is contained in~$S$. 
By a double counting argument, 
	one can easily see that for any~$k$-regular graph on~$n$ vertices, the product~$k \cdot n$ is even. Moreover, for all~$k$,~$n \in \mathbb{N}$ with~$k < n$ and~$k \cdot n$ being even, a circulant~$k$-regular graph on~$n$ vertices can be constructed in linear time. Indeed, to obtain a~$k$-regular graph in case~$k$ is even, 
	one must only ensure that~$|S| = k/2$. If~$k$ is odd (which implies that~$n$ is even), it suffices to have~$n/2 \in S$ and~$|S| = (k+1)/2$ 
	(see Figure~\ref{fig:circulant} for~$n=12$ and~$S = \{1,2,3,6\}$). 
We call this the \dfn{circulant construction}.

\begin{figure}[H]
	\subfloat[the circulant construction\label{fig:circulant}]{%
		\includegraphics[width=0.28\textwidth] {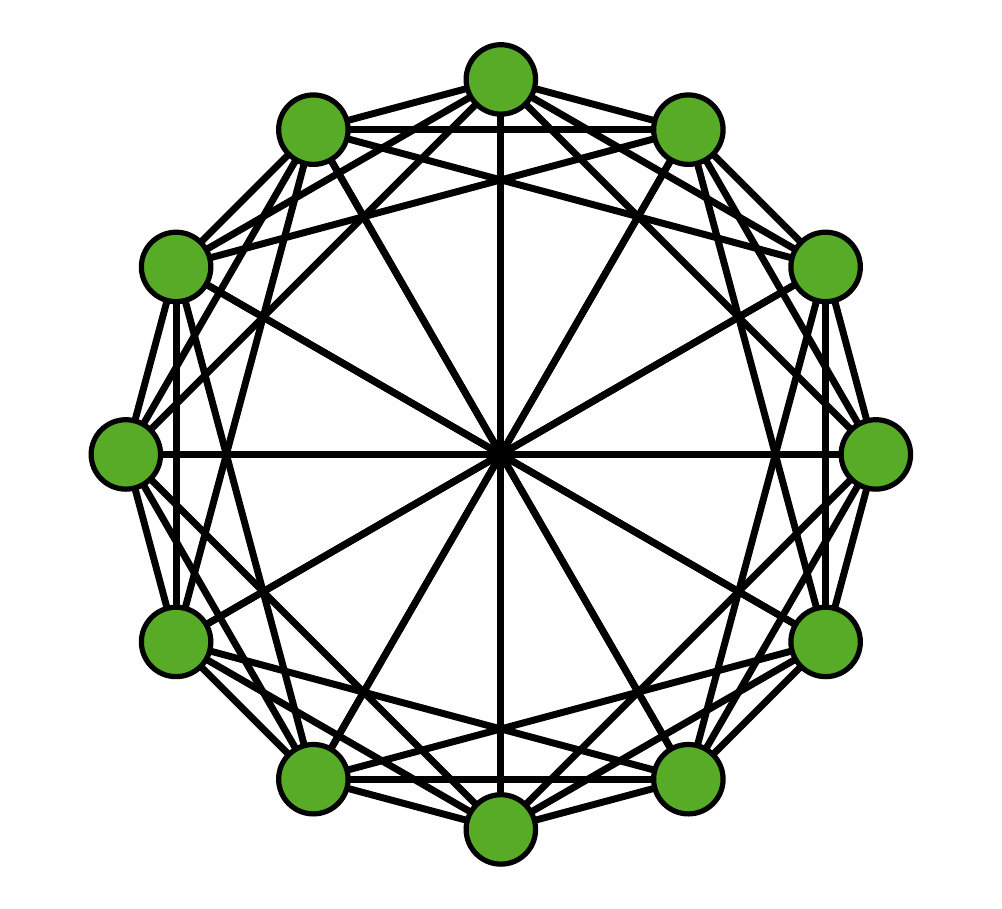}
	}
	\hfill
	\subfloat[the doubly-circulant construction\label{fig:doublyCirculant}]{%
		\includegraphics[width=0.28\textwidth] {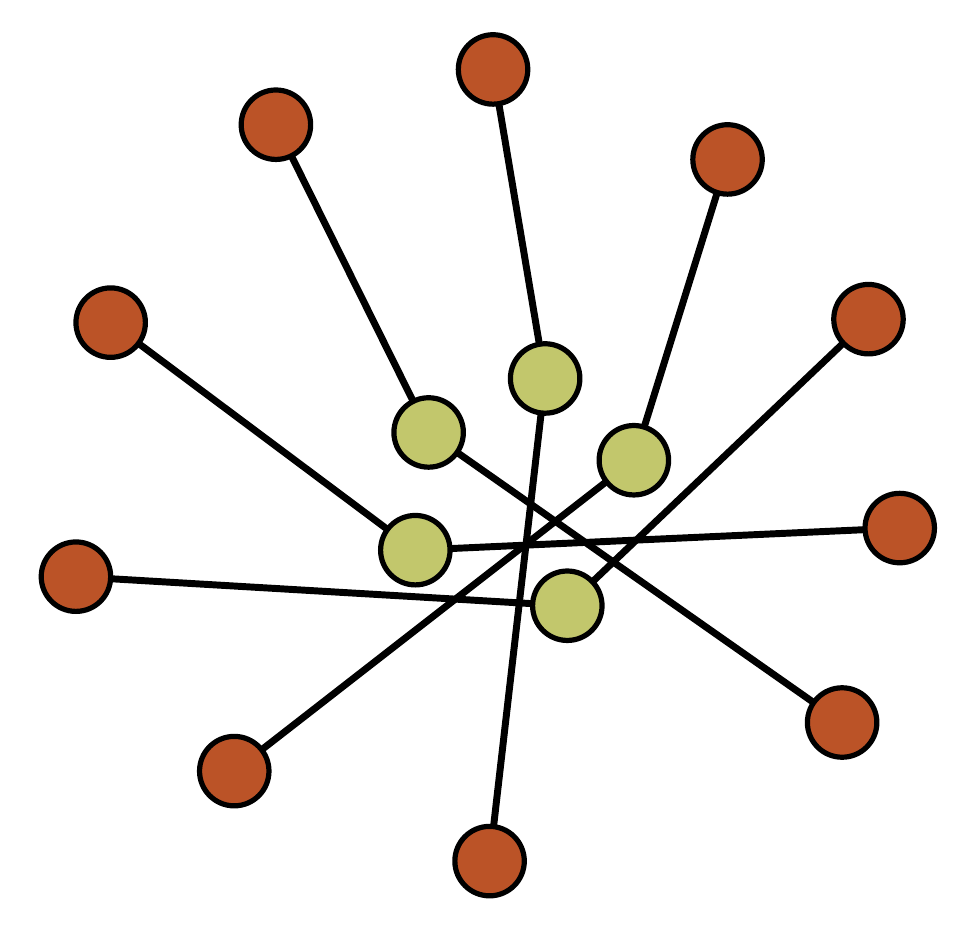}
	}
	\hfill
	\subfloat[Walecki's 1-factorization\label{fig:walecki}]{%
		\includegraphics[width=0.28\textwidth] {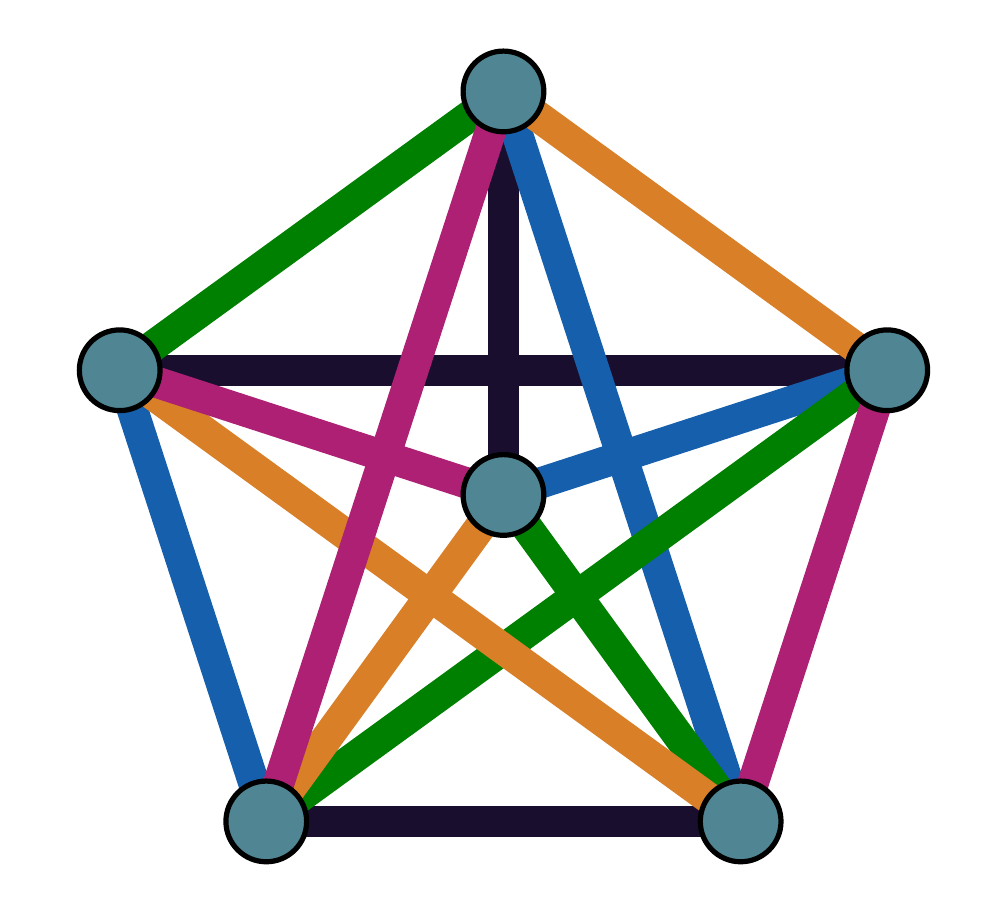}
	}
	\caption{}
	\label{fig:threeFigs}
\end{figure}

We have mentioned the necessary and sufficient evenness condition of~$k \cdot n$ for a~$k$-regular circulant graph on~$n$ vertices (with~$k < n$). For the existence of doubly-circulant graphs, we also obtain a simple criterion. 
Dy double counting, every~$(k,\ell)$-biregular graph 
	on a bipartition~$(P, Q)$ with~$\abs{P} = m$ and~$\abs{Q} = n$ satisfies~$k \cdot m = \ell \cdot n$.
The following lemma says 
	that this condition is essentially sufficient to guarantee the existence of a doubly-circulant~$(k,\ell)$-biregular graph and that such a graph can be constructed in linear time.

\begin{lemma}\label{lem_constructions}
	For all~$k$,~$\ell$,~$m$,~$n \in \mathbb{N}$ 
		with~$k \leq n$,~$\ell \leq m$ and~${k \cdot m = \ell \cdot n}$, 
		one can construct in linear time
		a~$(k, \ell)$-biregular graph on a bipartition~$(P,Q)$ 
		with~$\abs{P} = m$ and~$\abs{Q} = n$,
		which is doubly-circulant with respect to~$P$ and~$Q$. 
\end{lemma}

\begin{proof}
	Let~$P \coloneqq \{v_0, \dots, v_{m-1}\}$ and~$Q \coloneqq \{v'_0, \dots, v'_{n-1}\}$. We insert edges according to the following algorithm: Let~$j$ be the least integer such that the edge~$\{v_0, v'_j\}$ has not yet been inserted. (So in the beginning,~$j=0$.) For every~$v_i$, we insert edges to all vertices~$v'_{(i + j + s) \bmod n}$, where~$s$ ranges over all multiples of~$m$ up to~$\lcm(m, n)$. Thus, a vertex of~$P$ in the obtained graph has degree~$i_{P,Q} = \frac{\lcm(m,n)}{m}$. Let~$E_{P,Q}$ be the set of edges in the final graph we intend to construct. We have~$\abs{E_{P,Q}} = m \cdot k = n \cdot \ell$ edges, so~$\abs{E_{P,Q}}$ is a multiple of both~$m$ and~$n$. Thus,~$k = \frac{n \cdot \ell}{m} = \frac{\abs{E_{P,Q}}}{m} \geq \frac{\lcm(m,n)}{m} = i_{P,Q}$. If the inequality is proper, we can repeat the above process. In one iteration, the degree of a vertex in~$P$ increases by~$\frac{\lcm(m,n)}{m}$, a divisor of~$\lcm(m,n)$. Since~$\abs{E_{P,Q}}$ is a multiple of~$\lcm(m,n)$, we can therefore always ensure that the degree of a vertex in~$P$ never exceeds~$k$. It strictly increases with every iteration, so after a finite number of iterations we obtain a graph~$H$ in which a vertex in~$P$ has exactly~$k$ neighbors in~$Q$. 

	We provide an automorphism~$\varphi$ of~$H$ such that the orbits with respect to the cyclic group~$\langle\varphi\rangle$ generated by~$\varphi$ are~$P$ and~$Q$. Define~$\varphi$ such that it maps a vertex~$v_i \in P$ to~$v_{(i+1) \bmod m}$ and a vertex~$v'_j \in Q$ to~$v'_{(j+1) \bmod n}$. By construction,~$\varphi(P) = P$ and~$\varphi(Q) = Q$, so it remains to show that~$\varphi$ is a graph isomorphism, i.e., that~$\{\varphi(u),\varphi(v)\} \in E_{P,Q}$ if and only if~$\{u,v\} \in E_{P,Q}$ (this suffices since~$P$ and~$Q$ are independent sets). Let~$v_i \in P$,~$v'_j \in Q$. We have~$\varphi(v_i) = v_{(i+1) \bmod m}$ and~$\varphi(v'_j) = v'_{(j+1) \bmod n}$.

	Suppose~$\{v_i, v'_j\} \in E_{P,Q}$. Then~$j = i+r+s \mod n$ for some~$r \in \{0, \dots, m-1\}$ and some multiple~$s$ of~$m$ with~$s \leq \lcm(m, n)$. This implies that~$(j+1) = (i+1)+r+s \mod n$ for the same~$r$ and~$s$. Thus,~$\{\varphi(v_i), \varphi(v'_j)\} \in E_{P,Q}$. By a symmetric argument, we get equivalence. 
\end{proof}

The construction described in the proof is depicted in Figure~\ref{fig:doublyCirculant}.

The coarsest equitable partition of a graph is not necessarily its orbit partition. 
For example, on the disjoint union of two cycles of different sizes, 
	the coarsest equitable partition is the unit partition, 
	whereas the graph actually has two orbits. 
However, the above construction yields that for each~$\LogicCzwo$-equivalence class, 
	there is a representative
	whose coarsest equitable partition is indeed the orbit partition.

\begin{theorem}\label{thm_representative}
	For every graph~$G$, there is a~$\LogicCzwo$-equivalent graph~$H$
		which is multi-circulant with respect to its coarsest equitable partition.
\end{theorem}

\begin{proof}
	Let~$\{ P_1, \ldots, P_t \}$ be the coarsest equitable partition of~$G$.
	For every~$i\in \{1, \ldots, t\}$ the graph~$G[P_i]$ is regular.
	By the circulant construction, 
		there is a circulant graph~$H(P_i)$ on each~$P_i$, 
		which is~$\LogicCzwo$-equivalent to~$G[P_i]$.
	By Lemma~\ref{lem_constructions}, 
		we can connect~$H(P_i)$ to~$H(P_j)$ 
		in such a way that the resulting graph 
		on vertex set~$P_i\dunion P_j$ 
		is~$\LogicCzwo$-equivalent to~$G[P_i\cup P_j]$
		and doubly-circulant with respect to~$P_i$ and~$P_j$.
	In this way, by connecting the graphs~$H(P_i)$ to each other, 
		we get the desired graph~$H$.
\end{proof}

\begin{corollary}
	$\mathcal{I}_C^2$ admits linear time inversion on the class of graphs.
\end{corollary}

	Given an equivalence relation~$\equiv$ on a class~$\mathcal{C}$ of structures, the
		\dfn{canonization problem} for~$\equiv$
		is the problem of finding a map~$c\colon \mathcal{C} \to \mathcal{C}$ such that
		for every~$\mathfrak{A}$ in~$\mathcal{C}$ 
		we have~$c(\mathfrak{A}) \equiv \mathfrak{A}$ and
		for all~$\mathfrak{A}, \mathfrak{B}$ in~$\mathcal{C}$
		we have~$c(\mathfrak{A}) \equiv c(\mathfrak{B})$.
	The map~$c$ is called a \dfn{canonization} for~$\equiv$ 
		and~$c(\mathfrak{A})$ is called the \dfn{canon} of~$\mathfrak{A}$ (with respect to~$c$).
	Typically, the goal is to find such a canonization~$c$ that can be evaluated efficiently.

As a by-product of the described constructions, we obtain the following corollary.

\begin{corollary}\label{cor:canonization:complexity}
	Canonization of graphs for~$\LogicCzwo$
		can be done in time~$O((n+m) \log n )$.
\end{corollary}

Theorem~\ref{thm_representative} also immediately yields one of our main results.
\begin{corollary}\label{cor_orbit}
	If a graph~$G$ is identified by~$\LogicCzwo$, 
		then its coarsest equitable partition is the orbit partition.
\end{corollary}


Note that the inverse of Corollary~\ref{cor_orbit} does not hold: 
For example, on the graph consisting of two disjoint~$3$-cycles, 
	the coarsest equitable partition is the orbit partition, 
	but the graph is not identified by~$\LogicCzwo$ since it is~$\LogicCzwo$-equivalent to a~$6$-cycle.

In Section~\ref{sect_highdim}, we also show that Corollary~\ref{cor_orbit} 
	does not hold for any logic~$C^k$ with~$k\geq 3$.

\begin{remark}\label{remark:on:factorization:of:K6}
	It is natural to ask 
		whether Corollary~\ref{cor_orbit} holds for finite relational structures in general. 
	This is not the case, 
		not even for edge-colored undirected graphs. 
	For example, let the colors of the edges of the complete graph~$K_6$ 
		correspond to a~$1$-factorization (see Figure~\ref{fig:walecki}).
	There is only one~$1$-factorization of~$K_6$ up to isomorphism (see~\cite{Lindner1976265}).
	This graph is identified by~$\LogicCzwo$ 
		(see Theorem~\ref{thm:edge:colored:color:regular:graph:identified}). 
	It is rigid, i.e., its orbit partition is discrete, 
		whereas its coarsest equitable partition is the unit partition.
\end{remark}


\subsection{Inversion for finite structures}\label{subsec:finstruct:inv}

Let~$\mathfrak{A} = (A,R_1,\ldots,R_\ell)$ be a finite relational structure.
We define~$\left.\mathfrak{A}\right|_2$,
	the restriction of~$\mathfrak{A}$ to arity~$2$,
	to be the relational structure
	$(A,R'_1,\ldots,R'_\ell)$
	with
	\[
		R'_i \coloneqq 
			\{ (v_1, \ldots,v_{r_i})\in R_i 
				\mid \{v_1, \ldots,v_{r_i}\} \text{ has at most~$2$ elements}, 
			\}
	\]
	where~$r_i$ is the arity of~$R_i$.
Obviously, 
	two relational structures~$\mathfrak{A}$ and~$\mathfrak{B}$ are~$\LogicCzwo$-equivalent
	if and only if~$\left.\mathfrak{A}\right|_2$ and~$\left.\mathfrak{B}\right|_2$
	are~$\LogicCzwo$-equivalent.
Furthermore,~$\left.\mathfrak{A}\right|_2$ and~$\left.\mathfrak{B}\right|_2$ 
	are~$\LogicCzwo$-equivalent 
	if and only if~$\ecPOG(\left.\mathfrak{A}\right|_2)$ 
	and~$\ecPOG(\left.\mathfrak{B}\right|_2)$ are~$\LogicCzwo$-equivalent. 
Hence, the inversion problem for~$I^2_C$ on finite relational structures
	reduces to the inversion problem for~$I^2_C$ on~$\ecPOG$s.

The complete invariant~$I^2_C$ that we defined on the class of graphs 
	can naturally be extended to the class of~$\ecPOG$s:
we simply replace the matrix in the original definition of~$I^2_C$ with a matrix~$M$ such that for all~$i, j\in \{ 1, \ldots, k \}$
	the entry~$M_{ij}$ is a tuple encoding for each edge-color~$d$
	the number of~$d$-colored outgoing edges 
	from a vertex in~$P_i$
	to~$P_j$.

Similarly to the case of graphs,
	in order to solve the inversion problem for~$\ecPOG$s
	it suffices to solve it 
for the color-regular case and for the color-biregular case.

\paragraph{Color-regular case:}
Let~$n$ be the number of vertices. 
If~$n$ is odd, the degrees in each color in the underlying undirected graph must be even. Consequently, the circulant construction from Subsection~\ref{subsec:grap:inv}
	can be adapted to perform the inversion for directed and colored edges 
	(see Figure~\ref{fig:circulant-for-directed}). In more detail, for every map~$\psi \colon \{1,\ldots,n-1\} \rightarrow \{1,\ldots,t\}$ such that for all~$i$ we have~$\psi(n-i) = \psi(i)$ and for all~$j$ we have~$|\psi^{-1}(j)| = d_j$, we obtain a circulant graph~$\Circ(\psi)$ by coloring the edge between vertices~$i$ and~$j$ satisfying~$j\leq i$ with the color~$\psi(i-j)$. This can also be interpreted as a~$2$-factorization of the graph into~$(n-1)/2$ graphs that are 2-regular. The map~$\psi$ then dictates that the edges of the~$\ell$-th 2-factor are to be colored with color~$\varphi(\ell)$.

\begin{figure}[H] 
	\centering 
	{ \includegraphics[width=0.2\textwidth]{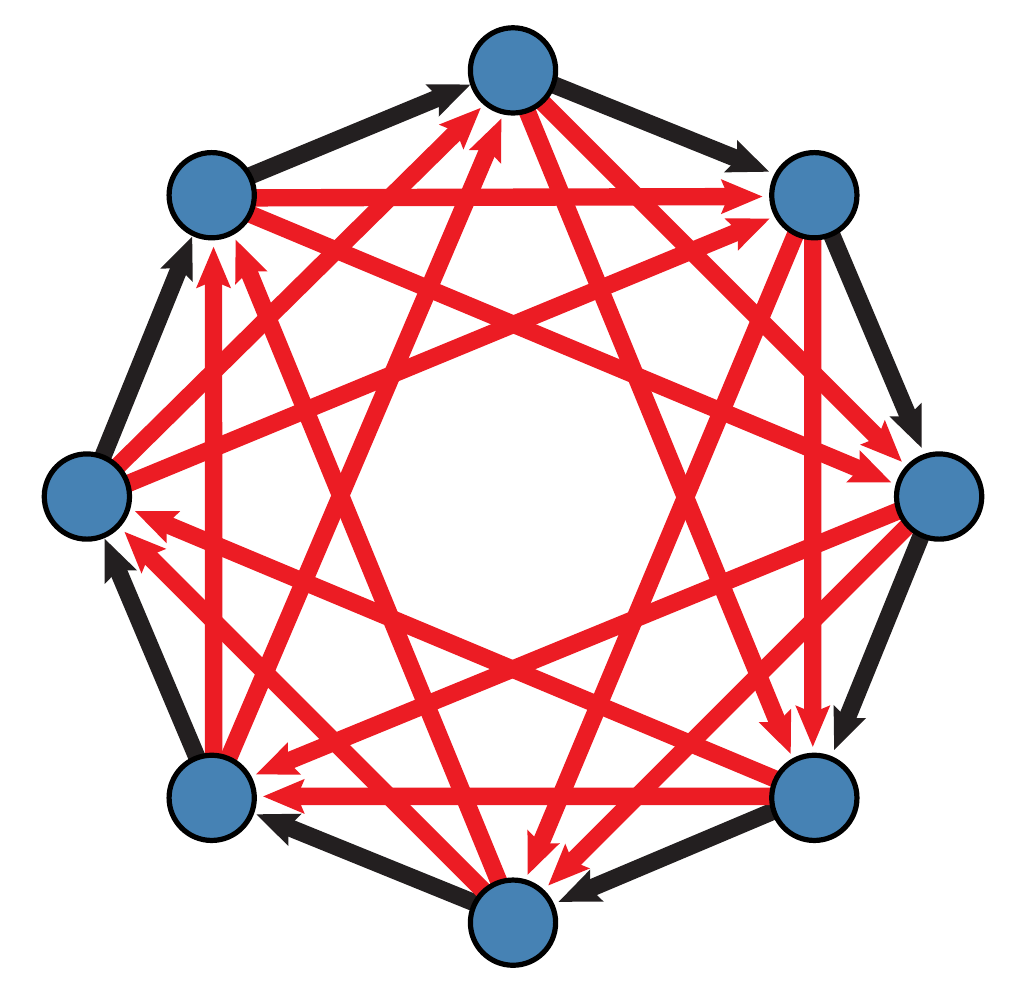} }
	\caption{
		The circulant construction can be adapted to colored graphs as long as all color degrees are even.
	}
	\label{fig:circulant-for-directed}
\end{figure}

If~$n$ is even and more than one color degree is odd, 
	we cannot apply the circulant construction. In fact, as Remark~\ref{remark:on:factorization:of:K6} shows it might not be possible to construct a transitive graph with the given color degrees.
Still, we can construct a canonical representative 
	for the corresponding~$\LogicCzwo$-equivalence-class as follows.
We use the 1-factorization construction due to Walecki (see~\cite{MR0124983}).
For even~$n$ the construction decomposes the complete graph~$K_n$ 
	into~$n-1$ disjoint perfect matchings. 
Geometrically, it is obtained as follows. 
Positioning the vertices~$1,\ldots,n-1$ in the plane to form a regular~$(n-1)$-gon and placing the vertex~$n$ into the center, each matching consists of an edge from~$n$ to some vertex~$i\in \{1,\ldots,n-1\}$ and all edges that are perpendicular to this edge~(see Figure~\ref{fig:walecki}).
In algebraic terms, the edge~$\{i,n\}$ is part of the~$i$-th matching for~$i\in\{1\,\ldots,n-1\}$. For~$i,j\in \{1,\ldots,n-1\}$ with~$i\leq j$ either~$|n-1-j+i|/2$ or~$|i+j|/2$ is an integer~$\ell$ and the edge~$\{i,j\}$ is part of the~$\ell$-th matching. This shows that the factorization can be computed in linear time
(see also \cite[Example 7.1.2.]{West2000}).

Suppose we want to have a graph with the specified color degrees~$d_1,\ldots,d_t$. By interpreting non-edges as edges of a special color, we can assume that~$\sum_{i = 1}^t d_i = n-1$. Then any function~$\psi \colon \{1,\dots,n-1\} \rightarrow \{1,\dots,t\}$ 
	with~$|\psi^{-1}(j)| = d_j$ yields an edge-colored color-regular graph 
	with the given color degrees 
	if the edges of the~$i$-th matching in the 1-factorization receive color~$\psi(i)$. 
Note that we can choose~$\psi$ canonically by mapping the first~$d_1$ integers to~$1$, 
	the next~$d_2$ integers to~$2$ and so on.  
	
We thus obtain a canonical construction, which yields, for every set of given color degrees, an edge-colored color-regular undirected graphs. 
We call this the \dfn{matching construction} 
	and denote for given color degrees~$d_1, \dots, d_t$ 
	the canon obtained from an edge-coloring function~$\psi$ by $\Match(\psi)$. 
The construction due to Walecki \cite{MR0124983}
		has the property that two matchings of the 1-factorization 
		always yield a Hamiltonian cycle. 
		If we require directed edges, in which case in-degrees must be equal to out-degrees, we pair a suitable number of matchings and orient the obtained Hamiltonian cycles.
With this  construction 
	it is directly possible to perform inversion 
	of~$\ecPOG$s
	and thus of finite relational structures in linear time.

\paragraph{Color-biregular case:}
The doubly-circulant construction that was described in Subsection~\ref{subsec:grap:inv} for graphs can directly be altered to handle colored directed edges, by treating each color one after the other.

\begin{corollary}
	$\mathcal{I}_C^2$ admits linear time inversion on the class of finite relational structures.
\end{corollary}


\section{\texorpdfstring{Characterization of the graphs identified by~$\LogicCzwo$}                   {Characterization of the graphs identified by~C2}}\label{sec:characterization:of:graphs}

Here we examine the graphs that are identified by the logic~$\LogicCzwo$ and give a complete characterization of them. We derive various results from the characterization.

\begin{definition}
	Let~$T$ be a tree with a designated vertex~$v$.
	For~$i\in\{1,\ldots,5\}$
		let~$(T_i,v_i)$ be an isomorphic copy of~$(T,v)$.
	Let~$F$ be the disjoint union of the five trees~$(T_i,v_i)$
		and~$E$ be the edge set of a~$5$-cycle on vertex set~$\{v_1, \ldots, v_5\}$.
	Then we call~$F+E$ a \dfn{bouquet}.

	A \dfn{bouquet forest} 
		is a disjoint union of 
		vertex-colored trees and 
		non-isomorphic  vertex-colored bouquets.
\end{definition}

\begin{figure}[H] \label{fig:bouquet-forest}
	\centering \includegraphics[width=0.9\textwidth]{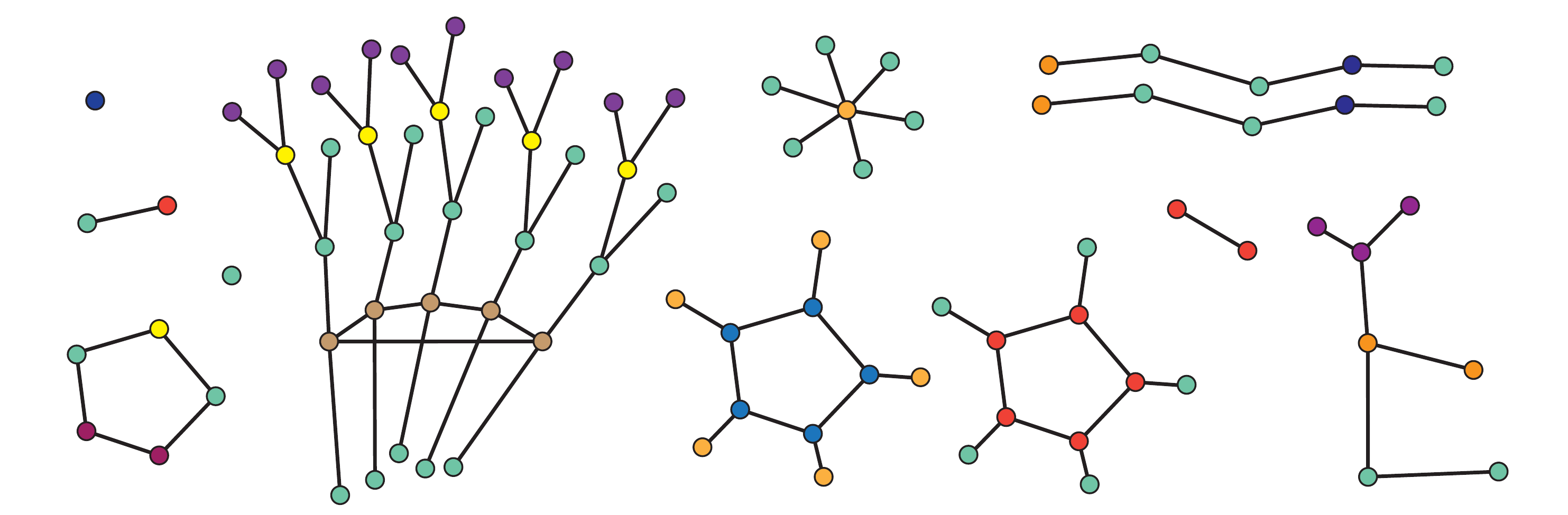}
	\caption{a bouquet forest with~$4$ bouquets}
\end{figure}

\begin{notation}
	For sets~$P$,~$Q$ we define~$[P,Q] \coloneqq {\{ \{v,w\} \mid v\neq w,\ v\in P,\ w\in Q\}}$.
\end{notation}

\begin{definition}
	Let~$G$ be a graph with~$\LogicCzwo$-coloring~$\chi$.
	We define the \dfn{flip of~$G$}
		to be the vertex-colored graph~$(F,\chi)$ with 
		$V(F) = V(G)$ and
		$E(F) = E(G) \symdiff ([P_1,Q_1]\cup \ldots \cup [P_t,Q_t])$, 
		where the~$(P_i,Q_i)$ are all pairs of (not necessarily distinct)~$\LogicCzwo$-partition classes of~$G$ which satisfy
		\[
			\abs{[P_i,Q_i]\cap E(G)} > \abs{[P_i,Q_i]\setminus E(G)}.
		\]
	We say that~$(F, \chi)$ is a \dfn{flipped graph}. 
		If~$\chi$ is the~$\LogicCzwo$-coloring of~$F$
			and~$(F, \chi)$ is a flipped graph,
			we also say that~$F$ is flipped.

\end{definition}

\bigskip
\begin{figure}[htb] \label{fig:a-graph-and-its-flip}
	\centering \includegraphics[width=0.7\textwidth]{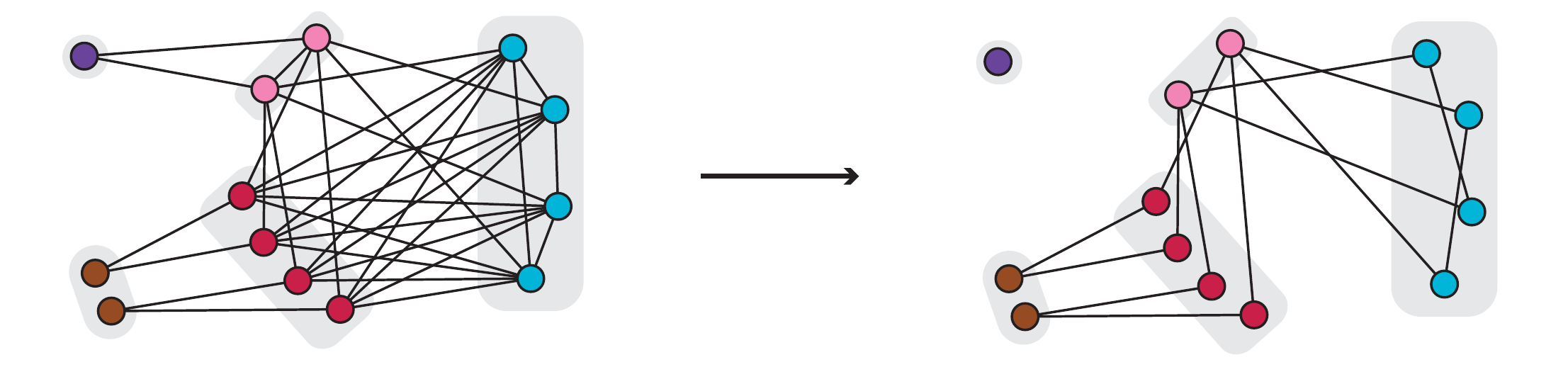}
	\caption{a graph with its~$\LogicCzwo$-partition and its flip}
\end{figure}

The symbol~$\Delta$ in the definition denotes, as usual, the symmetric difference.
To obtain the flip of~$G$, 
	we first calculate the~$\LogicCzwo$-partition of~$G$
	and check for each class~$P$
	whether~$G$ has more edges in~$P$ than non-edges.
If so, we take the complement of the edge relation there.
We do the same for all pairs~$P,Q$ of distinct~$\LogicCzwo$-partition classes,
	considering the edges in~$[P,Q]$ there.

While we may consider arbitrary  vertex colorings in general, we want to stress the fact that
	the coloring~$\chi$ of the flip~$(F,\chi)$ of a graph~$G$
	is always the~$\LogicCzwo$-coloring of~$G$.
Note that~$\chi$ induces an equitable partition of~$F$, 
	although it might not be the coarsest one.
The aim of this section is to prove 
	the following classification: a graph is identified by~$\LogicCzwo$
	if and only if its flip is a bouquet forest.
We first argue that we can restrict ourselves to flipped graphs.

\begin{lemma}\label{le:iff-identified}
	A graph~$G$ is identified by~$\LogicCzwo$ 
		if and only if its flip is identified by~$\LogicCzwo$.
\end{lemma}
\begin{proof}
	Let~$G$ be a graph and~$\chi$ its~$\LogicCzwo$-coloring.
	Let~$(F, \chi)$ be the flip of~$G$.
	Note that~$V(F) = V(G)$ by definition.
	Let~$G'$ be a graph on~$V(G)$
		with the same~$\LogicCzwo$-coloring~$\chi$ and let~$(F',\chi)$ be the flip of~$G'$.
	A map~$f: V(G)\to V(G)$ 
		is an isomorphism from~$(G,\chi)$ to~$(G',\chi)$
		if and only if it 
		is an isomorphism from~$(F,\chi)$ to~$(F',\chi)$. 
	This implies the contrapositives of both directions of the lemma.
\end{proof}

For a graph~$G$ and its~$\LogicCzwo$-coloring~$\chi$,
	we know that~$G$ and~$(G, \chi)$ have the same automorphisms. If for a~$\LogicCzwo$-partition class~$P$ of~$G$, 
	the induced subgraph~$G[P]$ is not identified by~$\LogicCzwo$, 
	then~$G$ is not identified by~$\LogicCzwo$.
The same is true for subgraphs 
	consisting of edges~$[P,Q]\cap E(G)$ for distinct classes~$P$ and~$Q$.
We are going to exploit this simple fact
	in order to prove the classification. 
The following two lemmas
	provide strong restrictions 
	to the edge relation of the regular and biregular graphs 
	identified by~$\LogicCzwo$.

\begin{lemma} \label{le:regular}
	A flipped regular graph is identified by~$\LogicCzwo$
		if and only if it is a graph with no edges, a matching, or a~$5$-cycle.
\end{lemma}

\begin{proof}
	The backward direction is easy to check.
 	For the forward direction, let~$G$ be a flipped~$k$-regular graph on~$n$ vertices.
	Assume that~$G$ has at least one edge and it is neither a matching nor a~$5$-cycle.
	Then we have~$n > 2$ and~$k \geq 2$.
	In order to show that~$G$ is not identified by~$\LogicCzwo$,
		we construct two non-isomorphic graphs~$H$ and~$H'$,
		both~$\LogicCzwo$-equivalent to~$G$.
	Note that the graphs which are~$\LogicCzwo$-equivalent to~$G$
		are exactly the~$k$-regular graphs on~$n$~vertices.

	Since~$k \geq 2$, we can apply the circulant construction from Subsection~\ref{subsec:grap:inv}
		in a way so that the resulting graph is connected.
	We choose~$H$ to be such a connected circulant graph.

	For the construction of~$H'$ we do a case analysis:
	\begin{enumerate}
		\item \emph{$k$ is odd:}
			Then~$n$ is even.
			\begin{enumerate}
			
			\item \emph{$n/2$ is even:}
					By assumption, we have~$k < n/2$. 
					Let~$H'$ be the disjoint union of 
						two~$k$-regular graphs on~$n/2$ vertices.
						
				\item \emph{$n/2$ is odd:} 
					Since~$k$ is odd and the graph is flipped, 
						we have~$2 < k < n/2 - 1$.
					Let~$H'$ be the disjoint union of two
						$k$-regular graphs on
						$n/2 + 1$ and~$n/2 - 1$ vertices, for example, circulant ones. 
			\end{enumerate}
			In both cases, we have~$H \not \cong H'$ since~$H'$ is disconnected.
		\item \emph{$k$ is even:}
			\begin{enumerate}
				\item \emph{$n$ is even:}
					Again we have~$2 \leq k < n/2$. 
					Since~$k$ is even, independent of the parity of~$n/2$, 
						there is a~$k$-regular graph on~$n/2$ vertices.
					We let~$H'$ be the disjoint union of two copies of such a graph.
					Again, we have~$H \not \cong H'$ since~$H'$ is disconnected. 
				\item \emph{$n$ is odd:}
					Since~$G$ is flipped, we have~$n\geq 2k + 1$.
					
					If~$n > 2k + 1$, from the parities of~$k$ and~$n$ we have~$k < n/2-1$ and
						we can choose 
						$H'$ to be the disjoint union of a~$k$-regular graph on~$\lfloor n/2 \rfloor$ vertices and a~$k$-regular graph on~$\lceil n/2 \rceil$ vertices.

					In case~$n = 2k + 1$, there is 
						no disconnected~$k$-regular graph on~$2k + 1$ vertices. Thus, in order to construct~$H'$,
						take two disjoint~$k$-cliques.
					Connect half of the first clique via a matching to half of the second clique.
					Lastly, connect all vertices in the unmatched halves
						to a new vertex~$v_0$. 

					Since~$G$ is not a~$5$-cycle by assumption, we cannot have~$k=2$.
					For~$k>2$, the vertex~$v_0$ is the only vertex which does not belong to a~$k$-clique.
					Hence,~$H'$ is not transitive, therefore not circulant.
					This implies~$H \not \cong H'$.
			\end{enumerate}
	\end{enumerate}
\end{proof} 

\begin{lemma}\label{le:biregular}
	The logic~$\LogicCzwo$ identifies a flipped~$(k,\ell)$-biregular graph
		if and only if~$k\leq 1$ or~$\ell\leq 1$. 
	Furthermore, 	
		if the graph is not identified,
		there is a non-isomorphic graph,
		which is~$(k,\ell)$-biregular on the same bipartition (even if~$k=\ell$).
\end{lemma}

\begin{proof}
For the backward direction note that
		if~$k\leq 1$ or~$\ell\leq 1$, 
		then~$G$ is a disjoint union of stars.
	Therefore,~$\LogicCzwo$ identifies~$G$.

	For the forward direction, let~$G$ be a flipped~$(k,\ell)$-biregular graph
		on bipartition~$(P ,Q)$.
	Without loss of generality assume that~$k \geq \ell$. Suppose~$k\geq 2$ and~$\ell\geq 2$.
	In order to show that~$\LogicCzwo$ does not identify~$G$,
		we construct non-isomorphic~$(k,\ell)$-biregular graphs~$H$ and~$H'$,
		both on bipartition~$(P,Q)$. We will define~$H$ and~$H'$ 
		such that~$H$ is connected and~$H'$ is disconnected.
	Let~$\abs{P} = m$ and~$\abs{Q} = n$.
	Note that our assumption~$k \geq \ell$ implies~$n\leq m$.
	
	Let~$H$ be a connected~$(k,\ell)$-biregular graph on bipartition~$(P,Q)$. Such a graph exists due to the following argument. Starting with~$G$, it suffices to explain how to repeatedly decrease the number of connected components while maintaining~$(k,\ell)$-biregularity. Since~$k>1$ and~$\ell>1$, each connected component has a cycle. Let~$C_1$ and~$C_2$ be two distinct connected components and let~$(p_1,q_1)$ and~$(p_2,q_2)$ be edges in these connected components, respectively, that lie on a cycle, where~$p_i\in P$ and~$q_i\in Q$. We replace the edges~$(p_1,q_1)$ and~$(p_2,q_2)$ by the edges~$(p_1,q_2)$ and~$(p_2,q_1)$. This decreases the number of connected components by~$1$ and maintains vertex degrees. We conclude that there is a connected~$(k,\ell)$-biregular graph on bipartition~$(P,Q)$.

	We now define the graph~$H'$.
	Suppose
		${P = \{v_0, \ldots, v_{m-1}\}}$ and
		${Q = \{w_0, \ldots, w_{n-1}\}}$.
	Let the complete bipartite graph on bipartition 
		${(\{v_0, \ldots, v_{\ell-1}\}, \{w_0, \ldots{}, w_{k-1}\})}$
		be one of the connected components of~$H'$.
	Since~$G$ is flipped, 
		we have~${k \leq n/2}$ and~${\ell \leq m/2}$.
	It follows that~${k\leq n-k}$ and~${\ell\leq m-\ell}$.
	By~${m \cdot k = n \cdot \ell}$, we have~${(m-\ell)\cdot k = (n-k)\cdot \ell}$.
	By Lemma~\ref{lem_constructions}, there is a~$(k,\ell)$-biregular graph 
		on bipartition
		${(\{v_\ell, \ldots, v_{m-1}\}, \{w_k, \ldots{}, w_{n-1}\})}$.
	We take such a graph as the rest of~$H'$.
	Obviously~$H'$, is disconnected.
	This completes the proof.
\end{proof}

\begin{notation}\label{notation:relations}
	For a~$(k,\ell)$-biregular graph on bipartition~$(P,Q)$ we introduce the following notations.
	\begin{align*}
		P \squarerel Q 	&:\iff k = \ell = 0\\
		P \doteq Q 	&:\iff k = \ell = 1\\
		P \ll Q 	&:\iff k \geq 2 \text{ and } \ell = 1
	\end{align*}
\end{notation}
We also write~$Q \gg P$ instead of~$P \ll Q$.

With this notation, Lemma~\ref{le:biregular} says that a flipped graph is identified by~$\LogicCzwo$ if and only if for every two distinct~$\LogicCzwo$-partition classes~$P$ and~$Q$ we have~$P \squarerel Q$,~$P \doteq Q$ or~$P \ll Q$.

For a graph~$G$ with~$\LogicCzwo$-partition~$\Pi$, we define the \dfn{skeleton~$S_G$} of~$G$ as the graph
	with~$V(S_G) = \Pi$
	and~$E(S_G) = \{ \{P,Q\} \mid P\doteq Q \text { or }  P\ll Q  \text{ in the flip of } G\}$.

\begin{lemma}\label{le:restr}
	Let~$G$ be a flipped graph which is identified by~$\LogicCzwo$. Then the following hold:
	\begin{enumerate}
\item there is no path~$P_0,P_1,\ldots,P_t$ in~$S_G$ with~$P_0\ll P_1$ and~$P_{t-1} \gg P_{t}$,\label{item:monotone:graph} 
\item there is no path~$P_0,P_1,\ldots,P_t$ in~$S_G$ where~$P_0\ll P_1$ and~$P_{t}$ induces a~$5$-cycle or a matching,\label{item:cannot:increase:to:execption:graph}
\item in every connected component of~$S_G$ 
	there is at most one vertex~$P$ 
	that induces a matching or a~$5$-cycle in~$G$.\label{item:at:most:one:exept:graph}
	\end{enumerate}
\end{lemma}
\begin{proof}
	For each condition, assuming it does not hold, we will define non-isomorphic graphs~$H$ and~$H'$,
		both~$\LogicCzwo$-equivalent to a subgraph of~$G$ 
		that is induced by some classes of the~$\LogicCzwo$-partition of~$G$.

	\begin{figure}[htb]
		\subfloat[\label{fig:two-stars}]
			{ \includegraphics[width=0.24\textwidth] {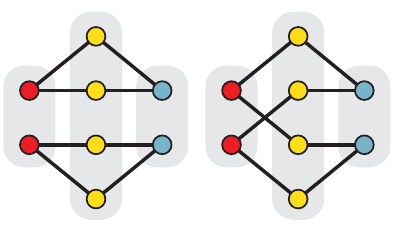}}
		\hfill
		\subfloat[\label{fig:matching-and-star}]
			{\includegraphics[width=0.24\textwidth] {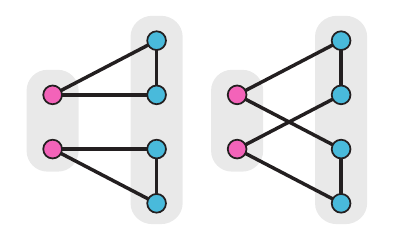}}
		\hfill
		\subfloat[\label{fig:two-five-cycles}]
			{\includegraphics[width=0.24\textwidth] {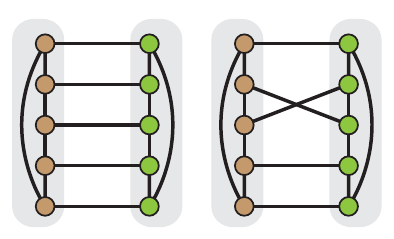}}  
		\hfill
		\subfloat[\label{fig:two-matchings}] 
			{\includegraphics[width=0.24\textwidth] {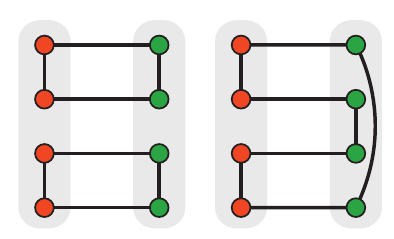}} 
		\caption{
			pairs of non-isomorphic graphs 
				illustrating the necessity of the conditions in Lemma~\ref{le:restr} 
		}
		\label{fig:proof-of-characterization-together} 
	\end{figure}

	\emph{Notation.}
	For two~$\LogicCzwo$-partition classes~$P$ and~$Q$,
		let~$i_{PQ}$ denote the number of neighbors that a vertex in~$P$ has in~$Q$.		

	(Part~\ref{item:monotone:graph})
First assume that~$t = 2$ 
			and~$G[P_1]$ has no edges.
		Assume that in~$G$ we have~$P_0\ll P_1 \gg P_2$.
		We may assume~$\abs{P_0}\leq \abs{P_2}$. 
		Let~$p\coloneqq i_{P_0P_1}$ and~$q\coloneqq i_{P_2P_1}$.
		Then~$p \geq q$.

		In the following,~$\Pi$,~$\Theta$ and their indexed variants
			will always denote partitions of~$P_1$
			into classes of size~$p$ and~$q$, respectively.
		Furthermore, we denote by~$H(\Pi,\Theta)$ the graph on vertex set~$P_0 \cup P_1 \cup P_2$ 
			with the following properties.
			\begin{itemize}
				\item For every $i\in \{ 0,1,2\}$ the subgraph induced by $P_i$ has no edges,
				\item we have~$P_0 \ll P_1 \gg P_2$,
				\item the classes of~$\Pi$ form the neighborhoods in~$P_1$ of vertices in~$P_0$,
				\item the classes of~$\Theta$ form the neighborhoods in~$P_1$ of vertices in~$P_2$.
			\end{itemize}

		First we define~$H$.
		Because~$p \geq q$, there are partitions~$\Pi_0$ and~$\Theta_0$ of~$P_1$
			such that there are at most~$\abs{P_0}-1$ classes of~$\Theta_0$
			that intersect two or more elements of~$\Pi_0$. 
		One way to construct them is to number the vertices in~$P_1$ 
			and then to define the $i$-th element of $\Pi_0$ (and $\Theta_0$) 
			as the set of the $i$-th $p$ (and $q$) vertices.
		We let~${H \coloneqq H(\Pi_0,\Theta_0)}$.

		We now define~$H'$. We know that~$\abs{P_0}\geq 2$ since~$G$ is flipped.
		Hence, since~$\abs{P_0} \cdot p = \abs{P_1} \cdot q$, there are partitions~$\Pi_1$ and~$\Theta_1$
			such that all 
			classes of~$\Theta_1$
			intersect two or more elements of~$\Pi_1$.
		We let~$H' \coloneqq H(\Pi_1,\Theta_1)$.
		Since~$|P_2| > |P_0|-1$, this implies~$H\not\cong H'$
		(see Figure~\ref{fig:two-stars}). This construction can easily be generalized for the case~$t>2$.

(Part~\ref{item:cannot:increase:to:execption:graph})
	Without loss of generality, we may assume that~$t = 1$.
	First let~$G[P_1]$ be a~$5$-cycle.
	Since~$P_0\ll P_1$ and~$5$ is prime, 
		the set~$P_0$ is a singleton.
	This contradicts the assumption that~$G$ is flipped.

	Let now~$G[P_1]$ be a matching.
	We let~$H$ be the graph~$G[P_1] + B$ on vertex set~$P_0 \cup P_1$,
		where~$B$ is a subset of~$[P_0,P_1]$
		such that we have~$P_0\ll P_1$ in~$H$
		and for each~$v\in P_0$ 
		the number of edges in~$G[P_1]$ 
		that~$v$ is connected to via~$B$
		is as small as possible.
	(If~$i_{P_0P_1}$ is even, this number is~$i_{P_0P_1} / 2$,
		otherwise it is~$(i_{P_0P_1}+1) / 2$.)

	We let~$H'$ be the graph~$G[P_1] + B'$,
		where~$B'$ is a subset of~$[P_0,P_1]$
		such that we have~$P_0\ll P_1$ in~$H'$
		and there is a vertex in~$P_0$
		which is connected to~$i_{P_0P_1}$ edges in~$G[P_1]$ via~$B'$.
	Such a~$B'$ exists because~$P_0$ cannot be a singleton.
	Figure~\ref{fig:matching-and-star} shows~$H$ and~$H'$ for~$t=1$ and~$\abs{P} = 4$.

(Part~\ref{item:at:most:one:exept:graph})
	Suppose~$P_0,P_1,\ldots,P_t$ is a path in~$S_G$ 
		such that both~$P_0$ and~$P_t$ induce a matching or a~$5$-cycle in~$G$. By Part~\ref{item:cannot:increase:to:execption:graph}, 
		we have~$P_0\doteq P_1 \doteq  \ldots \doteq  P_t$ in~$G$.
	By Lemma~\ref{le:regular}, without loss of generality, 
		we may assume that for every~$1 \leq i \leq t-1$, 
		the edge set of~$G[P_i]$ is empty.
	First, let~$G[P_0]$ be a~$5$-cycle. Then~$P_t$ also consists of~$5$ elements and cannot be a matching.
	Assume that~$G[P_t]$ is a~$5$-cycle.
	Let~$\tilde{H}$ be a graph on vertex set 
		$P_0 \dunion \ldots \dunion P_t$,
		such that for every~$i\in \{1, \ldots, t\}$ 
		we have~$P_{i-1} \doteq P_i$ in~$\tilde{H}$ and
		the graph~$\tilde{H}[P_i]$ has no edge.
	Let~$C$ be a~$5$-cycle on~$P_0$.
	Then there is a unique cycle~$C'$ on~$P_t$
		such that whenever~$\{v,v'\}$ is an edge in~$C$ 
		and~$w,w'\in P_t$ are the two unique vertices 
		connected via paths in~$\tilde{H}$ to~$v$ and~$v'$, respectively,
		then~$\{w,w'\}\in C'$.
	Let~$C''$  be a~$5$-cycle on~$P_t$ that is different from~$C'$.
	We define 
		$H  \coloneqq \tilde{H}+C+C'$ and
		$H' \coloneqq \tilde{H}+C+C''$. 
	Then~$H$ and~$H'$ are non-isomorphic
		(see Figure~\ref{fig:two-five-cycles} for~$t=1$).
	
	Assume now that~$G[P_0]$ is a matching.
	Note that~$\abs{P_0}$ must be even.
	Obviously,~$G[P_t]$ cannot be a~$5$-cycle, so assume it is a matching.
	Let~$M$ be a matching on~$P_0$ and let~$\tilde{H}$ be defined as above.
	Then there is a unique matching~$M'$ on~$P_t$
		such that whenever~$\{v,v'\}$ is an edge in~$M$ 
		and~$w,w'\in P_t$ are the two unique vertices connected via edges in~$\tilde{H}$ to~$v$ and~$v'$, respectively,
		then~$\{w,w'\}\in M'$.
	Let~$M''$ be a matching on~$P_t$ that is different from~$M'$.
	Because~$G$ is flipped,~$\abs{P_t} = 2$ is impossible.
	Hence, such a matching~$M''$ exists.
	We define 
		$H  \coloneqq \tilde{H}+M+M'$ and
		$H' \coloneqq \tilde{H}+M+M''$
		(see Figure~\ref{fig:two-matchings} for~$t = 1$ and~$\abs{P_0} = 4$).
\end{proof}

Now we have all the ingredients to prove the main theorem of this section.

\begin{theorem}\label{th:characterization}
	A graph is identified by~$\LogicCzwo$
		if and only if its flip is a bouquet forest.
\end{theorem}
\begin{proof}
	We first show the backward direction. 
	By Lemma \ref{le:iff-identified}, 
		it suffices to show 
		that every bouquet forest is identified by~$\LogicCzwo$.
	This can be done by induction.
	The base case is that
		disjoint unions of non-isomorphic vertex-colored~$5$-cycles 
		are identified by~$\LogicCzwo$. Note that, in the definition of a bouquet forest, 
		we do not allow isomorphic vertex-colored bouquets. (The reason for this is
		that~$\LogicCzwo$ cannot distinguish two~$5$-cycles from a~$10$-cycle.)
	For the inductive step,
		it suffices to observe that a graph is identified by~$\LogicCzwo$ 
		if all vertex-colored versions of the graph obtained by removing all leaves of the graph are identified.

	For the forward direction,
		let~$G$ be a graph identified by~$\LogicCzwo$.
	Let~$(F,\chi)$ be the flip of~$G$.
	By Lemma \ref{le:iff-identified}, 
		we know that~$(F,\chi)$ is identified by~$\LogicCzwo$.
	
	By Lemma \ref{le:regular},
		for every color class~$P$ of~$(F,\chi)$
		the graph~$F[P]$ either has no edges or is a matching or a~$5$-cycle.
	
	By Lemma \ref{le:biregular}, 
		for every two distinct color classes~$P$ and~$Q$ of~$(F,\chi)$
		we have~$P \squarerel Q$,~$P \doteq Q$ or~$P \ll Q$.
	
	We now show that the skeleton~$S_G$ is a forest.

	First let~$P_0, P_1, \ldots, P_t$ with~$t \geq 2$
		be color classes of~$(F,\chi)$ with~$P_0 \doteq P_1 \doteq \ldots \doteq P_t$.
	We claim that~$P_0 \doteq P_t$ is impossible.
	Assume~$P_0 \doteq P_t$ for the sake of a contradiction.
	Note that~$\abs{P_0} \neq 1$ because~$F$ is flipped.
	So we have~$\abs{P_0}\geq 2$.
	Then we can arrange the edges in the given matchings
		as~$\abs{P_0}$ disjoint cycles each of size~$t+1$,
		but also as a single cycle of size~$(t + 1) \cdot \abs{P_0}$.
	This implies that~$(F,\chi)$ is not identified by~$\LogicCzwo$.
	By Lemma \ref{le:iff-identified}, 
		we can conclude 
		that~$G$ is not identified by~$\LogicCzwo$, a contradiction.

	Let~$P_0,\ldots,P_t$ be a sequence of color classes
		such that for every~$i\in\{1,\ldots,t\}$ 
		we have~$P_{i-1} \ll P_i$ or~$P_{i-1} \doteq P_i$.
	Assume that~$P_{i-1} \ll P_i$ for at least one~$i\in\{1,\ldots,t\}$.
	Then by cardinality reasons 
		it cannot be the case that~$P_t \doteq P_0$ or~$P_t \ll P_0$.

	Also note that 
		if~$P_0 \ll \ldots \ll P_t$ with~$t\geq 2$, 
		then Part~\ref{item:monotone:graph} of Lemma~\ref{le:restr} implies 
		that~$P_0\ll P_t$ is impossible.

	We conclude that~$S_G$ is a forest.

	Furthermore, by Part~\ref{item:at:most:one:exept:graph} of  Lemma~\ref{le:restr},
		for every connected component~$T$ of~$S_G$
		there can be at most one class~$P\in V(T)$
		such that~$F[P]$ is a matching or a~$5$-cycle.
	To see that each connected component that contains a~$5$-cycle is a
	bouquet, note that by Corollary~\ref{cor_orbit}, all vertices of the~$5$-cycle must be in the same
	orbit and thus, the five trees obtained by deleting the~$5$-cycle must be isomorphic.
		With Part~\ref{item:cannot:increase:to:execption:graph} of Lemma~\ref{le:restr}, 
		it follows that the subgraph of~$F$ induced by the vertices of~$T$
	is either a forest or a bouquet. 
	Moreover, no~$\LogicCzwo$-partition class of~$F$ contains more than one~$5$-cycle,
		since this would contradict the fact that~$(F, \chi)$ is identified.
	Hence, different~$5$-cycles must have different colors. 
	This implies that there are no isomorphic colored bouquets in~$(F,\chi)$. 
\end{proof}

\begin{corollary}\label{cor:1}
	Given a graph with~$n$ vertices and~$m$ edges,
		we can decide whether it is identified by~$\LogicCzwo$ in time~${O((m+n)\log n)}$.
\end{corollary}
\begin{proof}
	Color refinement can be performed in time~$O((m+n)\log n)$ 
		(see~\cite{DBLP:conf/esa/BerkholzBG13}).
	Flipping a graph 
		and checking whether the flip is a bouquet forest
		can be done in linear time.
\end{proof}

A second corollary of Theorem~\ref{th:characterization}
	is concerned with vertex colorings of graphs that are identified by~$\LogicCzwo$.

\begin{corollary}\label{cor:refinement:of:ident:graph:is:ident}
	Let~$(G,\chi)$ be a vertex-colored graph  which is identified by~$\LogicCzwo$
	 	and let~$\chi'$ be a vertex coloring of~$G$ 
		which induces a finer partition on~$V(G)$ than~$\chi$ does.
	Then~$(G,\chi')$ is also identified by~$\LogicCzwo$.
\end{corollary}
\begin{proof}
	Note that the statement is true if~$G$ is a bouquet forest. 
	For the general case, we may assume that~$\chi$ and~$\chi'$ are equitable.   
	Let~$(F,\chi)$ be the flip of~$G$ and let~$(F',\chi')$ be the flip of~$(G,\chi')$. 
	Then the flip of~$(F,\chi')$ is equal to the flip of~$(F',\chi')$. 
	Since~$F$ is a bouquet forest, 
		we conclude that~$(F,\chi')$ is identified and thus, by applying Lemma~\ref{le:iff-identified} twice, we have that~$(F',\chi')$ is identified. Again by Lemma~\ref{le:iff-identified}, the vertex-colored graph~$(G,\chi')$ is identified.
\end{proof}

In the corollary, the vertex colorings~$\chi$ and~$\chi'$ of the graph~$G$ can be arbitrary, in particular~$\chi$ can be monochromatic. We show in Section~\ref{sect_highdim} that even when~$\chi$ is required to be monochromatic this result cannot be generalized to the logics~$\LogicCK$ for~$k>2$.

Since they can appear only once per connected component of the skeleton, we call a~$C^2$-partition class that is a~$5$-cycle or a matching an~\dfn{exception}. Our classification of finite relational structures that are identified by~$\LogicCzwo$, which we give in the next section, 
	depends on the structural properties proven about identified graphs in this section. 
We summarize them for convenience.

\begin{corollary}\label{cor:summary:of:pros}
	A flipped graph~$G$ is identified by~$\LogicCzwo$ if and only if the following hold:
	\begin{enumerate}
		\item 
			Each~$C^2$-partition class induces a graph identified by~$C^2$ 
				(i.e., the induced graph has no edges or it is a matching or a~$5$-cycle), 
		\item 
			for all~$C^2$-partition classes~$P$ and~$Q$ 
				we have~$P\squarerel Q$,~$P \doteq Q$,~$P \ll Q$ or~$Q \ll P$,
		\item 
			the skeleton~$S_G$ is a forest, 
		\item 
			there is no path~$P_0,P_1,\ldots,P_t$ in~$S_G$ 
				with~$P_0\ll P_1$ and~$P_{t-1} \gg P_{t}$ in~$G$,
		\item 
			there is no path~$P_0,P_1,\ldots,P_t$ in~$S_G$ 
				where~$P_0\ll P_1$ and~$G[P_{t}]$ is a~$5$-cycle or a matching, and
		\item 
			in every connected component of~$S_G$
				there is at most one exception (i.e., a class~$P$ that induces a matching or a~$5$-cycle).
	\end{enumerate}
\end{corollary}


\section{General finite structures}\label{sec:general:finite:structures}

Generalizing our observations about graphs, we now classify which finite relational structures are identified by~$\LogicCzwo$.

\textbf{Outline of the classification.} 
To achieve our classification we proceed as follows. 
As previously observed (see Subsection~\ref{subsec:rel:struct:and:pocs:prelims}), it suffices to analyze which edge-colored partially oriented graphs (i.e., $\ecPOG$s) are identified.  We would essentially like to follow the outline of the proof of the classification for graphs presented in Section~\ref{sec:characterization:of:graphs}. That is, we first intend to characterize
which color-regular~$\ecPOG$s are identified, and then which color-biregular~$\ecPOG$s are identified by $\LogicCzwo$. Note that color-regular~$\ecPOG$s are exactly the~$\ecPOG$s which have only one~$\LogicCzwo$-partition class. As we argue later, for an~$\ecPOG$ that is identified, the underlying undirected graph is also identified. Exploiting this observation we will first consider only undirected graphs. In our proof for the color-regular case, we need the undirected color-biregular  
case, so we treat it first. We then generalize the results to color-regular $\ecPOG$s and finally assemble these observations to classify all~$\ecPOG$s that are identified by~$\LogicCzwo$. 

Throughout this section we consider only complete~$\ecPOG$s. We can do this without loss of generality since we can interpret non-edges as edges of a particular color. The proofs of this section depend on the results from Section~\ref{sec:characterization:of:graphs}. 
We employ these results by interpreting uncolored graphs with edges and non-edges as complete graph with two edge colors.

In agreement with this outline we now start by analyzing color-biregular~$\ecPOG$s that do not have directed edges, i.e, undirected edge-colored color-biregular graphs. Recall that~$K_{3,3}$ is the complete bipartite graph with partition classes of size~$3$.

\begin{lemma}\label{lem:gen:math:structures:bipartite}
Let~$G$ be an undirected complete bipartite 
edge-colored color-biregular graph on bipartition~$(P,Q)$. Suppose~$G$ has at least three edge colors. Then there is a bipartite graph~$G'$ on bipartition~$(P,Q)$ which is non-isomorphic and~$\LogicCzwo$-equivalent to~$G$ if and only if~$G$ is not isomorphic to an edge-colored~$K_{3,3}$.
\end{lemma}

\begin{proof}
It is straightforward to check that for every regular edge coloring of~$K_{3,3}$ there is no~$\LogicCzwo$-equivalent bipartite non-isomorphic graph. We thus focus on the converse.

Let $G$ be undirected complete edge-colored color-biregular on bipartition~$(P,Q)$. Suppose~$|P| = p\cdot k$ and~$|Q| = q\cdot k$ with co-prime~$p$ and~$q$.

Let~$1\leq d_1^P\leq d_2^P \leq \dots \leq d_t^P$ and~$1\leq d_1^Q\leq d_2^Q \leq \dots \leq d_t^Q$ be the color degrees of colors~$1,\ldots,t$ for vertices in~$P$ and~$Q$, respectively. Since $t \geq 3$, it holds that~$\abs{P} \geq 3$ and~$\abs{Q} \geq 3$. By double counting, we have that~$q$ divides~$d_i^P$ and~$p$ divides~$d_i^Q$ for all~$i\in \{1,\ldots,t\}$.

(\emph{Case~$|P| = |Q|$}.) We first consider the case~$|P| = |Q|$. Note the following observation. If we construct a color-regular graph on~$(P,Q)$ for some subset of the colors with the correct color degrees, then it is always possible to complete the graph such that it has the correct color degrees for all edge colors. The reason is K\"{o}nig's Theorem, which implies that a regular bipartite graph contains a perfect matching.

Assume $d_1^P = d_2^P=1$ (which implies~$d_1^Q = d_2^Q = 1$). If $|P| = 3$, the graph is an edge-colored $K_{3,3}$. However, if~$|P|>3$, there are two non-isomorphic (not necessarily complete) bipartite graphs with edge colors~$1$ and~$2$ and color degrees~$d_1^P$ and~$d_2^P$: we can form a color-alternating Hamiltonian cycle or two shorter color-alternating cycles. By K\"{o}nig's Theorem, these graphs can be extended and we conclude that for a bipartite graph with~$|P| = |Q|>3$ and~$d_1^P = d_2^P =1$ there is a~$\LogicCzwo$-equivalent non-isomorphic graph on the same bipartition. Similarly, by Lemmas~\ref{le:iff-identified} and~\ref{le:biregular}, if~$1<d_1^P < |Q|-1$, then there are two non-isomorphic~$(d_1^P,d_1^Q)$-biregular graphs on bipartition~$(P,Q)$ which can both be extended. Note that since there are at least three colors,~$d_1^P< |Q|-1$, so this resolves the case~$|P| = |Q|$.

(\emph{Case~$|P| \neq |Q|$}.)  Suppose now that~$|P| \neq |Q|$. Consider a bipartite graph~$G'$ on $(P',Q')$ with~$|P'| = |Q'| = |P|/p = |Q|/q = k$ and with color degrees~$d_1^P/q, \ldots, d_t^P/q$ and~$d_1^Q/p, \ldots, d_t^Q/p$. From~$G'$ we obtain a graph with bipartition classes of sizes~$|P|$ and~$|Q|$ by replacing every vertex~$v$ in~$P'$ by~$p$ copies that each have in every color the same neighbors as~$v$ and then replacing each vertex~$w$ in~$Q'$ by~$q$ copies that each have in every color the same neighbors as~$w$.

If we perform the construction twice with non-isomorphic graphs~$G'_1$ and~$G'_2$ then we obtain two non-isomorphic~$\LogicCzwo$-equivalent graphs~$G_1$ and~$G_2$. We already know from the first part of the proof that such non-isomorphic graphs~$G'_1$ and~$G'_2$ exist if~$k\neq 3$.

It remains to consider the case~$k = 3$. Since we have already treated~$K_{3,3}$, we can assume without loss of generality that~$p>1$.
As above, we will construct two non-isomorphic graphs~$G'_1$ and~$G'_2$ where the second bipartition class $Q'$ has size~$|Q'| \coloneqq |Q|/q$, that is, $|Q'| = 3$, and then replace every vertex in that class by~$q$ copies with the same neighbors in each color in $P$. If~$G'_1$ and~$G'_2$ are non-isomorphic then the two final graphs are non-isomorphic. 
Recall that we assume that there are at least three edge colors. Since there are~$p k \cdot qk= 3 \cdot pqk$ edges in the complete bipartite graph and since there must be at least~$pqk = \lcm(pk,qk)$ edges in a color, there are exactly three colors, say red, green and blue, each with~$pqk$ edges.
Let~$Q' = \{w_0,w_1,w_2\}$.  
Consider the following two graphs~$G_1'$ and~$G_2'$. We divide~$P$ into three blocks~$P_0,P_1,P_2$ of equal size. We color all edges between $P_i$ and $w_i$ red, edges between $P_i$ and $w_{(i+1) \bmod 3}$ green and the other edges blue, as to form a blown-up color-regular bipartite~$K_{3,3}$. This yields the graph~$G'_1$. To obtain~$G'_2$ we choose for all~$i \in \{0,1,2\}$ a vertex $v_i$ from~$P_i$ and recolor its incident red edge green and its incident green edge red. That is, the edges~$\{v_i,w_i\}$ are now green while the edges~$\{v_i,w_{(i+1) \bmod 3}\}$ are now red.
The newly obtained graph~$G'_2$ is non-isomorphic to~$G'_1$ since in~$G'_2$ there are vertices that agree on some edge colors towards vertices in~$Q'$ but disagree on others. (For example, the edge $\{v_0, w_0\}$ is green just like the edges between vertices in~$P_2 \backslash \{v_2\}$ and~$w_0$, but~$\{v_0,w_1\}$ is red, whereas~$P_2$ only has blue edges to~$w_1$. Note that~$P_2 \backslash \{v_2\}$ is non-empty since $p > 1$) In~$G'_1$, for any vertex~$v \in P$, the color of any of its edges to~$Q'$ uniquely determines to which block of~$P$ it belongs and thus also determines the colors of the other edges incident to~$v$.
\end{proof}

In the lemma, the bipartition classes~$P$ and~$Q$ are not required to be distinguishable. 
However, if~$P$ and~$Q$ are a priori distinguishable (which is, for example, the case if vertices in~$P$ are colored with a different color than vertices in~$Q$), then the lemma also applies. The reason why the graphs~$G'_1$ and~$G'_2$ constructed in the proof are still~$\LogicCzwo$-equivalent to~$G$ in this situation, is that in case~$P$ and~$Q$ can be interchanged (which implies~$|P| = |Q|$), the color degrees for~$P$ and~$Q$ are the same. We obtain the following corollary.
\begin{corollary}\label{cor:gen:math:structures:bipartite:fixed:bipartition}
Let~$G$ be an undirected complete bipartite vertex-/edge-colored color-biregular graph on bipartition~$(P,Q)$. Suppose~$G$ has at least three edge colors, the vertices in~$P$ are colored red and the vertices in~$Q$ are colored blue. Then there is a bipartite graph~$G'$ on~$(P,Q)$ non-isomorphic to~$G$ and~$\LogicCzwo$-equivalent to~$G$ if and only if~$G$ is not isomorphic to~$K_{3,3}$.
\end{corollary}

In the following, we will need both Lemma~\ref{lem:gen:math:structures:bipartite} and Corollary~\ref{cor:gen:math:structures:bipartite:fixed:bipartition} depending on whether we are in a situation where the bipartition classes~$P$ and~$Q$ can be interchanged. We first use Lemma~\ref{lem:gen:math:structures:bipartite} to characterize color-regular graphs that are identified by~$\LogicCzwo$.

\begin{lemma}\label{lem:mixed:graphs:identified}
Let~$G$ be an undirected color-regular complete graph with at least three edge colors. Then~$\LogicCzwo$ identifies $G$ if and only if it is
\begin{enumerate}
	\item a graph on 4 vertices, which has $3$ edge colors that each induce a perfect matching,
	\item a graph on 6 vertices, which has 5 edge colors that each induce a perfect matching, or 
	\item a graph on 6 vertices, which has 3 edge colors, one of which induces the complement of a 6-cycle, and the other two edge colors each induce an undirected perfect matching. 
\end{enumerate}
\end{lemma}
\begin{proof}
Let~$1 \leq d_1\leq d_2\leq \ldots \leq d_t$ be the color degrees of the graph. Since we interpret non-edges as edges of a particular color, without loss of generality we can assume that~$\sum_{i = 1}^t d_i = n-1$.

We distinguish two cases according to the parity of the number~$n$ of vertices.

\emph{Case 1:~$n$ is odd.} 

Suppose first that $n$ is odd.
We will employ the generalized  circulant construction outlined in Subsection~\ref{subsec:finstruct:inv}. In particular, we saw there that for every map~$\psi \colon \{1,\ldots,n-1\} \rightarrow \{1,\ldots,t\}$ such that for all~$i$ we have~$\psi(n-i) = \psi(i)$ and for all~$j$ we have~$|\psi^{-1}(j)| = d_j$, we obtain a circulant graph~$\Circ(\psi)$ by coloring the edge between vertices~$i$ and~$j$ satisfying~$j\leq i$ with the color~$\psi(i-j)$.

Since $n$ is odd, all color degrees are even, and since there are at least three colors, we have $1<d_1\leq (n-1)/3$. This also implies~$n\geq 7$.

There are suitable~$\psi$,~$\psi'$ such that~$\psi^{-1}(1) = \{-d_1/2 ,\ldots,-1,1,\ldots,{d}_1/2\}$ and~$\psi'^{-1}(1) = \{-d_1/2 -1,\dots, -2,2,\dots,d_1/2 + 1\}$. We claim that~$\Circ(\psi)$ and~$\Circ(\psi')$ are not isomorphic.
Indeed, consider vertex~$0$. In~$\Circ(\psi)$ this vertex has two neighbors that both have~$d_1-2$ common neighbors with it (namely the neighbors~$1$ and~$n-1$). However, since~$d_1\leq (n-1)/3$, in~$\Circ(\psi')$ the vertex~$0$ does not have such a neighbor. Since the graph is transitive, it suffices to consider only vertex~$0$, and we conclude that the graphs are non-isomorphic.

\emph{Case 2:~$n$ is even.}  

Suppose now that~$n$ is even. Note first that all color-regular graphs on at most four vertices with at least three edge colors have exactly four vertices and are unions of three perfect matchings. We can thus assume in the following that~$n\geq 6$.

For every integer~$i$ define~$\even(d_i)$ to be the largest even integer that is at most~$d_i$. We distinguish two cases, depending on the following inequality, which we call the \emph{sum condition}. We say that the sum condition is fulfilled if~$\sum_{i = 2}^t \even(d_i) \geq n/2-1$. Note that this summation starts at index~$2$.

\emph{Case 2a: The sum condition is fulfilled.}

Observe that if the sum condition is fulfilled then there exists a graph~$G_H$ on~$n$ vertices indistinguishable from~$G$ with the following properties: The vertex set~$V$ of~$G_H$ is the union of two sets~$V_1$ and~$V_2$, each of size~$n/2$. The two graphs induced by~$V_1$ and~$V_2$ are isomorphic, color-regular graphs and the bipartite graph~$H$ induced by all edges that run between~$V_1$ and~$V_2$ is color-regular and contains all edges of color~$1$. This is possible by applying K\"onig's Theorem, since the summation in the sum condition starts at index~2.  We will, depending on the color degrees, construct such a graph~$H$ as to allow us to also construct a different graph~$H'$ that can replace~$H$ to obtain a non-isomorphic graph.

(\emph{$d_1>1$}.) If~$d_1 > 1$ then the graph~$H$ can be constructed in such a way that the graph induced by the edges in color~$d_1$ is connected (which can be seen by  starting with a Hamiltonian cycle and then applying K\"{o}nig's Theorem as in the previous proof). There are at least three colors and~$d_1$ is the smallest color degree, so we have~$d_1\leq (n-1)/3$. Since~$n\geq 6$, this implies that~$d_1 < n/2-1$ and thus, the graph induced by the edges of color~$1$ is neither a matching nor a co-matching. Lemma~\ref{le:biregular} implies that there is a bipartite graph~$H'$ that is~$\LogicCzwo$-equivalent but not isomorphic to~$H$. 
Let~$G_{H'}$ be the graph obtained from~$G_H$ by replacing~$H$ with~$H'$. Since the edges of color~$1$ in~$H$ induce a connected subgraph, the partition of~$G_{H}$ into the sets~$V_1$ and~$V_2$ is combinatorially determined by the graph~$G_{H}$ (i.e., the partition is isomorphism invariant). If this bipartition is not combinatorially determined by~$H'$ then the graphs~$G_H$ and~$G_{H'}$ are not isomorphic. However, if the bipartition is combinatorially determined by~$H'$ then~$G_{H}$ and~$G_{H'}$ are also not isomorphic since~$H$ and~$H'$ are not isomorphic. This resolves the case of the fulfilled sum condition if~$d_1> 1$.

(\emph{$d_1=d_2 = 1$}.) If~$d_1 = d_2 = 1$, we 
can construct~$G_H$ such that all edges of color~$1$ and~$2$ are contained in~$H$.
We repeat the argument above, this time constructing a bipartite graph~$H$ for which the subgraph induced by the edges of color~$1$ and the edges of color~$2$ is connected.
According to Lemma~\ref{lem:gen:math:structures:bipartite} there is exactly one size for which there exists no non-isomorphic graph~$H'$ that is~$\LogicCzwo$-equivalent to~$H$, namely the case~$n= 6$. Note that in this case, since the sum condition is fulfilled, we have~$d_3 = 3$ (which implies~$t = 3$), or it holds that~$d_3 = 1$ and~$d_4= 2$ (and~$t=4$). In the first case, since two disjoint matchings always form a union of cycles of even length, for every graph with these degrees, the graph induced by the third color must be
the complement of a~$6$-cycle and thus, the graph~$G$ is identified by~$\LogicCzwo$. In the second case, two non-isomorphic~$\LogicCzwo$-equivalent graphs are obtained by forming three matchings and two 3-cycles and by forming three matchings and a 6-cycle.

(\emph{$d_1=1$ and~$d_2 >1$}.) Supposing now that~$d_1 = 1$ and~$d_2 >1$, consider first the case~$n=6$, which implies~$d_2= d_3=2$. Such a graph is not identified since it is possible to let color~$d_2$ induce two triangles but also possible to let it induce a~$6$-cycle. Supposing~$n>6$ one can choose the graphs induced by~$V_1$ and~$V_2$ so as to have at least two edge colors. 
Indeed, since there are at least three edge colors and $d_3 \geq 2$, we can take a number~$d'_2$ with~$1\leq d'_2\leq d_2$ and let the induced color-regular graphs on $V_1$ and~$V_2$ be at least~$d'_2$-regular in color~2 and at least~$2$-regular in color 3, for example, by using the circulant construction from Subsection~\ref{subsec:grap:inv}.
The graph~$H$ contains a perfect matching of color~$1$. Note that this does not necessarily mean that the bipartition~$(V_1,V_2)$ is combinatorially determined by the graph.
Since the induced graphs on~$V_1$ and~$V_2$ are isomorphic, the matching can be chosen such that whenever~$(v_1,v_2)$ and~$(v_1',v_2')$ are distinct edges from~$V_1$ to~$V_2$ of color~$1$ then the color of~$(v_1,v_1')$ is equal to the color of~$(v_2,v_2')$. However, since~$G_H[V_2]$ has at least two edge colors, by suitably changing the matching, 
we can destroy this property maintaining the~$\LogicCzwo$-equivalence class.

We claim that this alteration produces a non-isomorphic graph. We argue this by showing that the number of 4-tuples~$(v_1,v_2,v'_1,v'_2)$ of distinct vertices such that~$(v_1,v_2)$ and~$(v'_1, v'_2)$ are matching edges and~$(v_1,v'_1)$ has a different color than~$(v_2,v_2')$ changes. Since the bipartition might not be isomorphism invariant, in addition to the edges considered above, we also need to consider for an unordered pair of matching edges~$(\{v_1,v_2\},\{v_1',v_2'\})$ with~$v_i,v'_i\in V_i$ whether~$(v_1,v_2')$ and~$(v'_1,v_2)$ are of the same color (note the reversal of the indices). However,~$(v_1,v_2')$ and~$(v'_1,v_2)$ are edges in the graph~$H$. Thus, the number of such pairs for which we obtain edges of different colors does not change if the matching is rearranged while maintaining the isomorphism type of~$H$. We can thus create two~$\LogicCzwo$-equivalent graphs which differ in the number of pairs of (directed) matching edges of color~$1$ whose starting vertices are connected by the same color as the end vertices. These graphs cannot be isomorphic.

\emph{Case 2b: The sum condition is not fulfilled.}

Suppose now that the sum condition is not fulfilled. Observe that this 
implies that~$d_1= d_2 = 1$, since otherwise~$\even(d_2)+\even(d_3) \geq (d_1+d_2+d_3+1)/2-1$ and~$\even(d_i) \geq d_i/2$ for all~$i\geq 3$. 
If there are at most four colors, then~$\sum_{i = 2}^t \even(d_i) \geq n-4$.
Hence, if~$n\geq 6$, the sum condition is fulfilled. 

We can thus assume that there are at least 5 colors and that~$d_1 = d_2 = 1$.
Every map~$\psi\colon \{1,\ldots,n-1\} \rightarrow \{1,\ldots,t\}$ leads to an edge-colored graph~$\Match(\psi)$ via the matching construction from Subsection~\ref{subsec:finstruct:inv}. Note that if~$\psi_1$ and~$\psi_2$ are two distinct maps then the graphs~$\Match(\psi_1)$ and~$\Match(\psi_2)$ are not identical (but may be isomorphic).

We will in the following consider only maps~$\psi$ that satisfy~$\psi(1) = 1$ and~$\psi(2) = 2$. 

There are at least $\binom{n-3}{d_3}\binom{n-3-d_3}{d_4}\geq (n-3) \cdot 2$ distinct valid choices for~$\psi$ that are of interest to us.
Since there are only~$n$ possible automorphisms of a color-alternating Hamiltonian cycle on~$n$ vertices, if all graphs resulting from the choices for~$\psi$ were isomorphic then we would have~$(n-3) \cdot 2\leq n$, which implies~$n\leq 6$.
We already considered~$n\leq 4$ so suppose~$n = 6$. In this case, since there must be~5 colors,~$d_1 = d_2 = d_3 = d_4 = d_5 = 1$. We thus obtain a 1-factorization of~$K_6$, of which there is only one up to isomorphism~\cite{Lindner1976265}, so the graph is identified by $\LogicCzwo$.
\end{proof}

So far we have only dealt with undirected graphs. We now extend the previous lemma
to color-regular complete $\ecPOG$s (see Figure~\ref{fig:the:identified:graphs}).
Recall that~$\ecPOG$s have the property that there is no edge color with both directed and undirected edges. (With this assumption we do not lose generality since directed edges and indirected edges are always distinguishable.)

\begin{theorem}\label{thm:edge:colored:color:regular:graph:identified}
Let~$G$ be a color-regular complete~$\ecPOG$. 
Then~$\LogicCzwo$ identifies $G$ if and only if~$G$ is 
\begin{enumerate}
\item an undirected complete graph with only one edge color, \label{item:undirected:complete}
\item undirected and has two edge colors, one of which induces a perfect matching, \label{item:perfect:matching}
\item a directed 3-cycle,\label{item:directed:3:cycle}
\item an undirected graph on four vertices, which has three edge colors that each induce a perfect matching,
\item a graph on four vertices, has two edge colors, one of which induces a 4-cycle that may or may not be directed, and the other edge color induces an undirected perfect matching,\label{item:4:cy:and:match}
\item the unique regular tournament on five vertices,
\item a graph on five vertices, has two edge colors which both induce a $5$-cycle of which at most one is directed,\label{item:5:cycle}
\item an undirected graph on six vertices, which has five edge colors that each induce a perfect matching, or 
\item an undirected graph on six vertices, which has three edge colors, one of which induces the complement of a~$6$-cycle, and the other two edge colors each induce an undirected perfect matching.\label{item:1:fac:of:K6}
\end{enumerate}
\end{theorem} 

\begin{figure}[H]
	\centering \includegraphics[width=\textwidth]{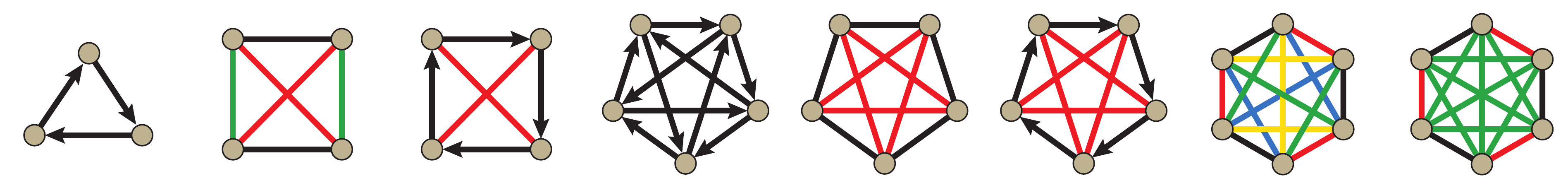} 
	\caption{
		The special cases (Items~\ref{item:directed:3:cycle}--\ref{item:1:fac:of:K6}) that occur in the classification
			in Theorem~\ref{thm:edge:colored:color:regular:graph:identified}.}\label{fig:the:identified:graphs}
\end{figure}
 \begin{proof}
Let~$G$ be an edge-colored color-regular complete graph.

Let~$\overline{G}$ be the graph that is obtained from~$G$ by replacing every directed edge with an undirected edge of the same color. 
We claim that if~$G$ is identified by $\LogicCzwo$, then so is~$\overline{G}$.
Indeed, assume there is a non-isomorphic graph~$\overline{H}$ not distinguished from~$\overline{G}$ by~$\LogicCzwo$. Then we obtain a non-isomorphic graph~$H$ not distinguished from~$G$ by directing all edges in edge color classes that consist of directed edges in~$G$. This is consistently possible (i.e., making the graph~$H$ color-regular) because each such edge color class induces a regular graph of even degree. By the classical result of Euler we can choose a Eulerian tour, which gives us equal in- and out-degrees for all vertices.

To prove the theorem it thus remains to consider graphs whose undirected underlying graph is identified by $\LogicCzwo$.

\emph{(at most 2 edge colors)} Starting with graphs that have at most two edge colors we consider all undirected graphs that appear in Lemma~\ref{le:regular} (with non-edges interpreted as edges of a distinct color). For graphs with only one edge color, the question is for which~$n$ there exists exactly one regular tournament on~$n$ vertices. This is the case for~$n\in \{1,3,5\}$. Indeed, to see that for~$n>5$ there is no unique regular tournament it suffices to consider the circulant directed graph in which there is an edge from~$i$ to~$j$ if~$(i-j)\bmod n <n/2$. Reversing all edges of the form~$(i,i+1)$ yields a non-isomorphic regular tournament.

Next we consider graphs that consist of a matching and the complement of a matching. Suppose the matching is of color red and the complement is blue. Since a matching induces a graph of odd degree, it cannot be directed while maintaining regularity. If~$n = 4$ then by directing the complement of the matching we obtain the graph described in Item~\ref{item:4:cy:and:match}. Note that $n$ is even, so suppose~$n>5$ from now on. We form two non-isomorphic graphs~$\widehat{G}$ and~$\widehat{G'}$ on~$\{v_1,\ldots,v_{n/2}\} \cup \{w_1,\ldots,w_{n/2}\}$ that are~$\LogicCzwo$-equivalent to~$G$ as follows.
We start with an arbitrary directed complete graph~$H$ on vertex set~$\{1,\ldots,n/2\}$ that contains the edges~$(1,2)$,~$(2,3)$ and~$(3,1)$. In~$\widehat{G}$ each pair~$\{v_i,w_i\}$ forms an edge of the red matching. For every directed edge~$(i,j)$ in~$H$, we insert the blue directed edges~$(v_i,v_j)$,~$(v_j,w_i)$,~$(w_i,w_j)$ and~$(w_j,v_i)$. This yields a graph~$\widehat{G}$ which consists of a red matching and a blue directed complement of a matching.
Note that in~$\widehat{G}$ every~$4$-tuple of vertices that induce two red matching edges also induces a directed 4-cycle in blue. Also note that~$\widehat{G}$ contains a blue directed 3-cycle~$C_3$ with vertices~$v_1,v_2,v_3$. We form a new graph~$\widehat{G'}$ by reversing all edges in $C_3$. The graph~$\widehat{G'}$ is~$\LogicCzwo$-equivalent to~$\widehat{G}$, but it is not isomorphic to~$\widehat{G}$ since the~$4$-tuple of vertices~$(v_1,v_2,w_1,w_2)$ induces two red matching edges but no blue directed~$4$-cycle in~$\widehat{G'}$.

Consider now a 5-cycle. Interpreted as an edge-colored complete graph, it has two edge colors, both of which induce a 5-cycle. If we direct one of the 5-cycles, we obtain a graph 
described in Item~\ref{item:5:cycle}. It is easy to see that these graphs are identified by $\LogicCzwo$ since reversing all edges of the~5-cycle yields an isomorphic graph. However, directing both 5-cycles, i.e., directing edges of both edge colors, we obtain a graph that is not identified by $\LogicCzwo$ since it is not isomorphic to the graph obtained when reversing the directions in one of the 5-cycles, since only one of the graphs has a directed~$3$-cycle with two edges in the first color.

\emph{(at least three edge colors)} Concerning graphs with three or more edge colors we need to consider all graphs mentioned in Lemma~\ref{lem:mixed:graphs:identified}. Since matchings cannot be oriented maintaining regularity it suffices to consider graphs on 6 vertices that are the disjoint union of the complement of a 6-cycle and two perfect matchings. However, for the edge color class that is the complement of a 6-cycle the degree is odd, so it cannot be oriented in a regular way either.
We thus obtain only the undirected versions of the graphs mentioned in Lemma~\ref{lem:mixed:graphs:identified}.
\end{proof}

We are going to determine the structure of~$\ecPOG$s that are identified by~$\LogicCzwo$. So far, we have classified the $\ecPOG$s that have exactly one~$\LogicCzwo$-partition class. Similarly to the case of undirected graphs treated in Section~\ref{sec:characterization:of:graphs} we will now describe how to combine such building blocks to form a larger graph that is identified. As before, by interpreting non-edges as edges of a special color we only need to consider complete graphs. Since vertices of different~$\LogicCzwo$-partition classes are distinguished, we can assume that all edges between two~$\LogicCzwo$-partition classes are undirected. (In other words, the direction of a directed edge starting in a color class~$P$ and ending in another color class~$Q$ is implied by the colors of the starting and end vertices. Edges of the same color in the other direction can thus always be marked with a different color). 
Recall that Corollary~\ref{cor:gen:math:structures:bipartite:fixed:bipartition} says that the only regular graph with a fixed bipartition and at least three edge colors that is color-biregular is~$K_{3,3}$  with an edge coloring that induces three disjoint perfect matchings.
This shows that it is possible that two~$\LogicCzwo$-partition classes of size~$3$ in an identified graph can be connected via three matchings of different colors.
In the light of this, in comparison to the case of two colors covered in Lemma~\ref{le:biregular} we need an additional relation. 

\begin{definition}
Let~$G$ be a vertex-colored~$\ecPOG$ and let~$P$ and~$Q$ be two disjoint subsets of~$V(G)$. We introduce the relation~$P \equiv_3^3 Q$ to denote the fact that the graph induced by the edges running between~$P$ and~$Q$ is the graph~$K_{3,3}$ with three edge colors which each induce a perfect matching between~$P$ and~$Q$.
\end{definition}

Since we treat non-edges as edges of a particular color, we also need to adapt our notation which we used for graphs (see Notation~\ref{notation:relations}) so that it becomes applicable to graphs with two edge colors. Suppose between two distinct~$\LogicCzwo$-partition classes~$P$ and~$Q$ there are edges of two colors, say red and blue.
To define the relations~$P\squarerel Q$,~$P \doteq Q$, $P \ll Q$ we always consider the graph induced by the edge color class that contains fewer edges and ignore orientations. It is not difficult to see that if the number of red edges is equal to the number of blue edges then choosing either induced graph yields the same results. Note that with this convention the relations~$P\squarerel Q$,~$P \doteq Q$, $P \ll Q$ in particular imply that there are only at most two colors among the edges running between~$P$ and~$Q$.

\begin{corollary}\label{cor:relationship:between:classes:in:iden:pog}
Let~$G$ be a vertex/edge-colored undirected graph that is identified by~$\LogicCzwo$. If~$P$ and~$Q$ are distinct $\LogicCzwo$-partition classes of~$G$, then~$P\squarerel Q$,~$P \doteq Q$, $P \ll Q$,~$Q \ll P$ or~$P \K33 Q$.
\end{corollary}

Theorem~\ref{thm:edge:colored:color:regular:graph:identified} implies that each~$\LogicCzwo$-partition class of an~$\ecPOG$ that is identified must be one of~$9$ listed types. In an identified~$\ecPOG$, we call every $\LogicCzwo$-partition class which does not induce an undirected complete graph with only one edge color an \dfn{exception} (i.e., any class that does not fall under Item~\ref{item:undirected:complete} of Theorem~\ref{thm:edge:colored:color:regular:graph:identified}). Similarly, we call every pair of $\LogicCzwo$-partition classes~$P$ and~$Q$ for which~$P\K33 Q$ holds an \emph{exception}. 

Note that with this terminology Lemma~\ref{le:restr} says that for uncolored, undirected graphs each connected component of the skeleton can induce at most one exception in the original graph. Here, in accordance with Lemma~\ref{le:regular}, the possible exceptions are Item~\ref{item:perfect:matching} and the undirected option in Item~\ref{item:5:cycle} from Theorem~\ref{thm:edge:colored:color:regular:graph:identified}.
As we will see, a similar observation about exceptions only occurring once per connected component of the skeleton is true for edge-colored graphs.

Since in an edge-colored graph it is not clear what a flip is, we need to slightly adjust the approach from Section~\ref{sec:characterization:of:graphs}.
Let~$G$ be a vertex/edge-colored graph. As before, we define the vertices of the \emph{skeleton}~$S_G$ to be the~$\LogicCzwo$-partition classes of~$G$. Two distinct vertices~$P$,~$Q$ in~$S_G$ are adjacent in~$S_G$ if the corresponding classes in~$G$ do not satisfy~$P\squarerel Q$, i.e., whenever~$P$ and~$Q$ are not monochromatically connected.

Concerning the properties of identified structures, we obtain a theorem similar to Corollary~\ref{cor:summary:of:pros} for graphs.

\begin{theorem}\label{thm:classification:of:finite:ident:sturctures}
Let~$G$ be a vertex-colored~$\ecPOG$ (i.e., a vertex/edge-colored partially oriented graph). Then~$G$ is identified by~$\LogicCzwo$ if and only if the following hold: 
\begin{enumerate}
\item Each~$\LogicCzwo$-partition class induces a graph identified by~$\LogicCzwo$ (i.e., the induced graph is one of the graphs mentioned in Theorem~\ref{thm:edge:colored:color:regular:graph:identified}),\label{item:classes:are:exceptions}
\item for all~$\LogicCzwo$-partition classes~$P$ and~$Q$ we have~$P\squarerel Q$,~$P \doteq Q$, $P \ll Q$,~$Q \ll P$ or~$P \K33 Q$, \label{item:bipartite:exceptions}
\item the skeleton $S_G$ is a forest,  \label{item:skeleton:is:forest}
\item there is no path~$P_0,P_1,\ldots,P_t$ in~$S_G$ with~$P_0\ll P_1$ and~$P_{t-1} \gg P_{t}$, \label{item:facing:stars}
\item there is no path~$P_0,P_1,\ldots,P_t$ in~$S_G$ where~$P_0\ll P_1$ and~$P_{t}$ is an exception, and \label{item:stars:to:exception}
\item in every connected component of~$S_G$ there is at most one exception (i.e., either two $\LogicCzwo$-partition classes~$P$ and~$Q$ satisfying~$P\K33 Q$ or a $\LogicCzwo$-partition class among Items~\ref{item:perfect:matching}--\ref{item:1:fac:of:K6} described in Theorem~\ref{thm:edge:colored:color:regular:graph:identified}).\label{item:at:most:one:exep} \end{enumerate}
\end{theorem}

\begin{proof}
For the backward direction let~$G$ first be a vertex-colored~$\ecPOG$ that satisfies the given conditions. It suffices to show that for each connected component in the skeleton~$S_G$, the graph induced by its collection of~$\LogicCzwo$-partition classes is identified by~$\LogicCzwo$. For such an induced graph note that, if the entire graph itself is not already an exception, then it contains
a non-exceptional $\LogicCzwo$-partition class~$P$ such that for exactly one class~$Q\neq P$ we have~$P\ll Q$ or~$P\doteq Q$ and for all other $\LogicCzwo$-partition classes~$Q'\neq Q$ (including~$P$ itself) we have~$P\squarerel Q'$.
Such a graph is identified if and only if the graph obtained by removing~$P$ is identified.
 It thus suffices to show that all exceptions are identified. This is shown in Corollary~\ref{cor:gen:math:structures:bipartite:fixed:bipartition} and~Theorem~\ref{thm:edge:colored:color:regular:graph:identified}. 

For the forward direction let~$G$ be a vertex-colored~$\ecPOG$ identified by $\LogicCzwo$. Note that the vertex coloring of~$G$ only has the effect that the~$\LogicCzwo$-partition is refined. 
Corollary~\ref{cor:summary:of:pros} and Theorem~\ref{thm:edge:colored:color:regular:graph:identified} imply Part~\ref{item:classes:are:exceptions}. 
Corollary~\ref{cor:summary:of:pros} and
  Corollary~\ref{cor:gen:math:structures:bipartite:fixed:bipartition} imply Part~\ref{item:bipartite:exceptions}.

Note the following for two~$\LogicCzwo$-partition classes~$P$ and~$Q$ of an identified graph that satisfy~$P \K33 Q$. If we merge two of the three edge color classes into a single one to obtain a graph~$G'$, then the entire graph must still be identified by~$\LogicCzwo$. The reason is that if two classes~$P$ and~$Q$ of size~3 satisfy~$P \doteq Q$ then the larger edge color class can always be split so that~$P\K33 Q$.
Similarly, if we merge different edge color classes within a $\LogicCzwo$-partition class~$P$ so that the graph induced by~$P$ is still identified then the entire graph must still be identified by~$\LogicCzwo$. Likewise, if we remove directions from some color classes of edges within~$P$ so that the graph induced by~$P$ is identified then the entire graph must still be identified by~$\LogicCzwo$.

In particular, we can replace every $\LogicCzwo$-partition class by an undirected complete graph while maintaining that the graph is identified. 

With this in mind, Parts~\ref{item:skeleton:is:forest}--\ref{item:stars:to:exception} follow immediately from Corollary~\ref{cor:summary:of:pros}.

To prove Part~\ref{item:at:most:one:exep} we first ignore the~$\K33$ exceptions. 
Note that by Parts~\ref{item:facing:stars} and~\ref{item:stars:to:exception} exceptions of the other types that occur in the same component of~$S_G$ must have the same size. Moreover, these exceptions must be connected by a path~$P_0, P_1, \ldots, P_t$  such that we have $P_0 \doteq P_1 \doteq \ldots \doteq P_t$ in~$G$. 
Suppose in such a situation~$P_0$ and~$P_t$ each induce an exception.
Note that for all such exceptions apart from the tournaments on~$3$ and~$5$ vertices, it is possible to remove directions and merge colors in order to obtain a matching or a 5-cycle. It follows from Corollary~\ref{cor:gen:math:structures:bipartite:fixed:bipartition} that such exceptions cannot occur in the same component of the skeleton.
If~$P_0$ and~$P_t$ are directed $3$-cycles, then by reversing the directions in~$P_t$ we obtain a non-isomorphic graph. Similarly, if~$P_0$ is the regular tournament on~$5$ vertices and~$P_t$ is an arbitrary exception on~$5$ vertices with some directed edges then we obtain a non-isomorphic graph by reversing all edges in~$P_0$. If~$P_t$ is an exception on~$5$ vertices with no directed edges then~$P_t$ must be the exception that has two edge colors which each induce an undirected~$5$-cycle. Since~$P_0$ contains~$5$-cycles of the underlying undirected graph for which the edges form a directed cycle and~$5$-cycles for which the edges do not form a directed~$5$-cycles, by permuting the vertices in~$P_t$ we can construct a non-isomorphic graph (by either choosing to match a~$5$-cycle of~$P_t$ with a directed~$5$-cycle of~$P_0$ or with a~$5$-cycle whose edges do not form a consistent cyclic orientation).

Suppose finally that an exception of the type~$\K33$ occurs in a connected component of~$S_G$. Then by merging colors and by Parts~\ref{item:facing:stars} and~\ref{item:stars:to:exception} the only types of exceptions that can occur in the connected component are other exceptions of type~$\K33$ and directed~$3$-cycles.  
Observe that an exception of type~$P \K33 Q$ induces a cyclic order on both~$P$ and~$Q$. Thus, by swapping two edge colors of the~$\K33$-exception we obtain a non-isomorphic graph whenever there is a second exception.
\end{proof}

\begin{corollary}
Given a vertex-colored~$\ecPOG$ on~$n$ vertices we can decide  whether it is identified by~$\LogicCzwo$ in time~$O(n^2 \log n)$.
\end{corollary}

\begin{proof}
For this we compute the coarsest equitable partition in time~$O(n^2 \log n)$ using the color refinement procedure. We then compute the skeleton of the obtained graph. For each~$\LogicCzwo$-partition class we check whether the induced graph appears among the exceptions from~Theorem~\ref{thm:edge:colored:color:regular:graph:identified}. Since the non-trivial exceptions have at most~6 vertices, this can be performed in linear time. We also check whether each pair of distinct~$\LogicCzwo$-partition classes satisfies one of the relations required by Corollary~\ref{cor:relationship:between:classes:in:iden:pog}. We mark each vertex or edge of the skeleton that corresponds to an exception. Finally, we check whether the skeleton is a tree, whether each connected component has at most one exception, and whether the sizes of the~$\LogicCzwo$-partition classes in each connected component monotonically increase when starting from a smallest class (for example using depth first search). We also check that every exception is a smallest class of its connected component.
\end{proof}

Since we can reduce the problem of deciding whether a finite relational structure is identified by~$\LogicCzwo$ to the problem of deciding whether an~$\ecPOG$ is identified by~$\LogicCzwo$ we obtain the following corollary.

\begin{corollary}\label{cor:runtime:structure} 
Given a finite relational structure~$\mathfrak{A}$ with a universe of size~$n$ over a fixed signature we can decide in time~$O(n^2 \log n)$ whether it is identified by~$\LogicCzwo$. 
\end{corollary}
\begin{proof}
Since the logic~$\LogicCzwo$ can only use two variables, if~$n \geq 3$ and~$\mathfrak{A}$ has a relation of arity at least~3, then it is not identified. For a fixed signature and a given structure~$\mathfrak{A}$ over this signature, the graph~$\ecPOG(\mathfrak{A})$ can be constructed in~$O(n^2)$ time.
It suffices now to observe that a relational structure~$\mathfrak{A}$ with relations of arity at most~$2$ is identified if and only if~$\ecPOG(\mathfrak{A})$ is identified.
\end{proof}

Using the classification we also obtain an extension of Corollary~\ref{cor:refinement:of:ident:graph:is:ident} to~$\ecPOG$s and, more generally, to finite relational structures.

\begin{corollary}
	Let~$G$ be an~$\ecPOG$ and~$\chi$ a vertex coloring of~$G$ such that~$(G,\chi)$ is identified by~$\LogicCzwo$
	and let~$\chi'$ be a vertex coloring of~$G$ 
		which induces a finer partition on~$V(G)$ than~$\chi$ does.
	Then~$(G,\chi')$ is also identified by~$\LogicCzwo$.
\end{corollary}

\begin{proof}
By induction, it suffices to consider a vertex coloring that refines exactly one~$\LogicCzwo$-partition class of~$G$.
First note that every vertex-colored version of an undirected complete graph has no exception. Furthermore, for a graph that is a perfect matching, if we consider the $\LogicCzwo$-partition classes obtained after coloring the vertices arbitrarily, then each component of the skeleton induces at most one exception (which must be a matching).

Next note that every vertex-colored version of a graph described in Items~\ref{item:directed:3:cycle}--\ref{item:1:fac:of:K6} of Theorem~\ref{thm:edge:colored:color:regular:graph:identified} is identified and has at most one exception. 
This is easy to see for all described graphs up to~$5$ vertices.
For the mentioned graphs on~$6$ vertices consider first the case of~$5$ disjoint matchings. It is not difficult to see that this graph only has two equitable partitions, the discrete partition and the trivial partition. Consider now the graph on~$6$ vertices that consists of two perfect matchings and a~$3$-regular graph. Every equitable partition
with a singleton for this graph is discrete. 
There are thus two types of equitable partitions to consider, namely equitable partitions with three classes of size~$2$ and equitable partitions with two classes of size~$3$. In the latter case, the graph induces an edge-colored version of~$K_{3,3}$ and in the former case the graph has at least two classes~$P_1$ and~$P_2$ such that~$P_1 \squarerel P_2$ (since each of the two matchings can only appear in one biregular connection between two classes). In any case, the graph is identified and has at most one exception. This proves that Properties \ref{item:classes:are:exceptions} and~\ref{item:at:most:one:exep} of Theorem~\ref{thm:classification:of:finite:ident:sturctures} are preserved.

To show globally that a color-refined version of an identified~$\ecPOG$ is also identified, note that also all other properties required by~Theorem~\ref{thm:classification:of:finite:ident:sturctures} are preserved, since (by Corollaries~\ref{cor:refinement:of:ident:graph:is:ident} and~\ref{cor:summary:of:pros}) these properties are preserved if all $\LogicCzwo$-partition classes are replaced by undirected complete graphs.
\end{proof}

\begin{corollary}
If a finite relational structure~$\mathfrak{A}$ is identified by~$\LogicCzwo$, then every finite relational structure obtained from~$\mathfrak{A}$ by adding unary relations is also identified by~$\LogicCzwo$.
\end{corollary}

\begin{proof}
Since~$\mathfrak{A}$ is identified, it has at most~$2$ elements in its universe or there is no relation of arity at least~$3$. Recall that under this circumstance,~$\mathfrak{A}$ is identified if and only if~$\ecPOG(\mathfrak{A})$ is identified. The addition of unary relations corresponds to a refinement of the vertex coloring of~$\ecPOG(\mathfrak{A})$, so the corollary follows from the previous one.
\end{proof}

\section{Higher dimensions}\label{sect_highdim}

One may wonder whether the results presented for~$\LogicCzwo$ throughout this paper can be extended to the logics~$\LogicCK$. Chang~\cite{chang} and Hoffman~\cite{hoffman1960} showed that triangular graphs on sufficiently many vertices are identified by~$\LogicCdrei$. They form an infinite family of strongly regular graphs, which shows that for~$k>2$ any classification result for graphs identified by~$\LogicCK$ must include an infinite number of non-trivial graphs. This already indicates that the situation for~$k=2$ is special.
We first show that statements analogous to Corollaries~\ref{cor_orbit} and~$\ref{cor:refinement:of:ident:graph:is:ident}$ do not hold for higher dimensions.
We use the construction from~\cite{DBLP:journals/combinatorica/CaiFI92} which provides us with the following fact.
\begin{fact}[Cai, F\"urer, Immerman~\cite{DBLP:journals/combinatorica/CaiFI92}] \label{fact:CFI}
For every~$k>2$, there are non-isomorphic~$3$-regular graphs~$G$ and~$G'$ of size~$O(k)$ on the same vertex set~$V$, which contains vertices~$a$,~$b$,~$a'$,~$b'$ with the following properties: The graphs are not distinguishable by~$\LogicCK$ and the edge sets satisfy~$E(G)\setminus E(G') = \{\{a,a'\},\{b,b'\}\}$
and~$E(G')\setminus E(G) = \{\{a,b'\},\{b,a'\}\}$. Moreover, we can guarantee that every graph indistinguishable from~$G$ and~$G'$ by~$\LogicCdrei$ is isomorphic to~$G$ or~$G'$.
\end{fact}

The graphs~$G$ and~$G'$ referred to in the fact are obtained by replacing every vertex~$v$ of a suitable connected graph by a gadget that only depends on the degree of~$v$ (see~\cite{DBLP:journals/combinatorica/CaiFI92} for details). Each original edge~$\{v,v'\}$ corresponds to two disjoint edges in the new graph, connecting a pair of specified vertices~$\{a_v, b_v\}$ in the gadget for~$v$ to a pair of specified vertices~$\{a_{v'}, b_{v'}\}$ in the gadget for~$v'$. If the two disjoint edges are~$\{a_v,a_{v'}\}$ and~$\{b_v, b_{v'}\}$, we call them \emph{parallel}. Otherwise, if they are~$\{a_v,b_{v'}\}$ and~$\{b_v, a_{v'}\}$, we call them \emph{twisted} and the pair of them a \emph{twist}.

\begin{figure}[H]
	\centering \includegraphics[width=0.4\textwidth]{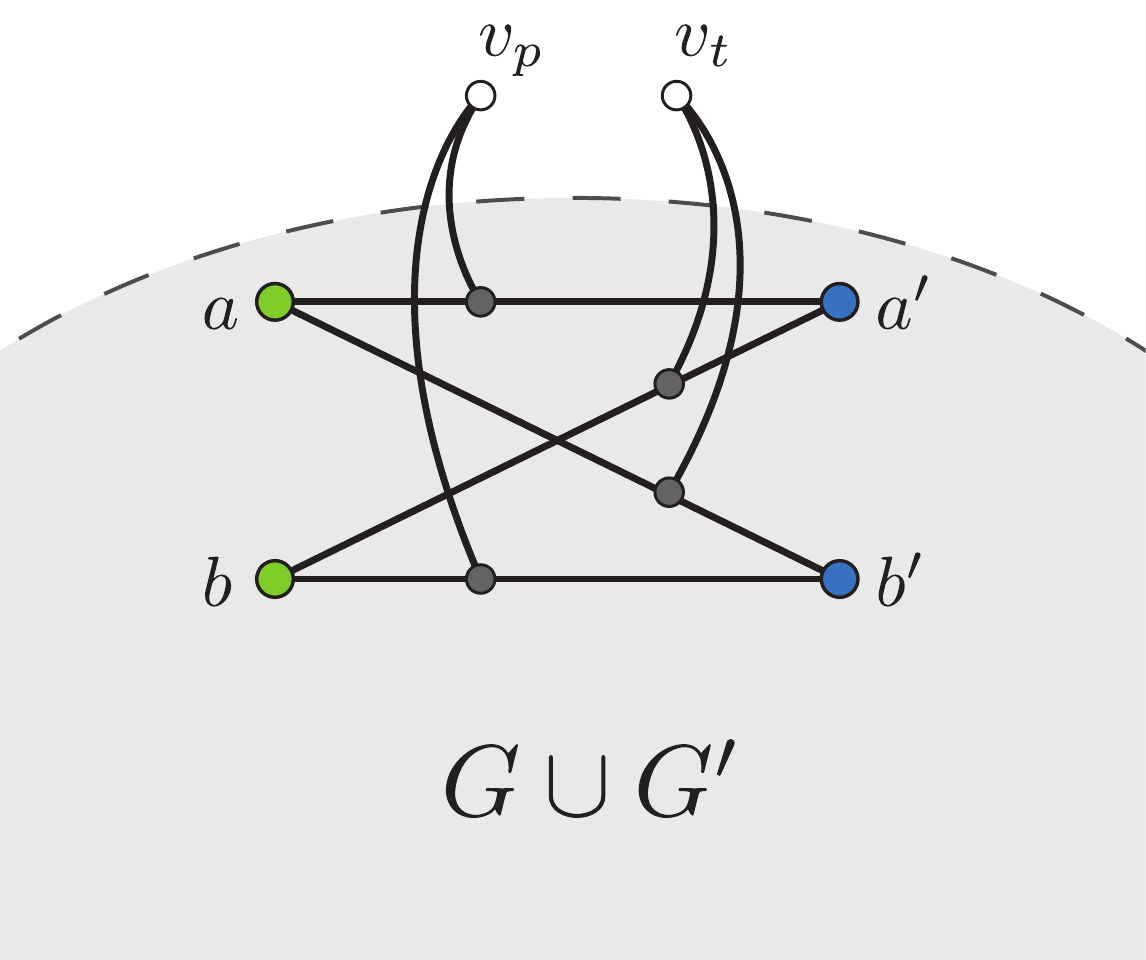}
	\caption{the construction used to show that higher-dimensional analogues of Corollaries~\ref{cor_orbit} and~\ref{cor:refinement:of:ident:graph:is:ident} do not hold} 
\end{figure}

Each gadget consists of so-called middle vertices and~$d$ pairs of outer vertices, where $d$ is the degree of the corresponding vertex in the original graph. Cai, F\"urer and Immerman show the following. When choosing each pair of edges between gadgets of adjacent vertices to be arbitrarily either parallel or twisted, we can only obtain one of two graphs up to isomorphism. More precisely, the isomorphism type of the resulting graph only depends on the parity of the number of twists in it. 

In~\cite{DBLP:journals/combinatorica/CaiFI92}, first colored graphs are constructed. By a standard transformation which reduces isomorphism of colored graphs to isomorphism of uncolored graphs, one can remove the colors in the graph. Since there are suitable transformations that do not add more than~$O(k)$ extra vertices, we can assume the graphs~$G$ and~$G'$ appearing in Fact \ref{fact:CFI} to be uncolored, still preserving the property that for a fixed input graph, all possible resulting graphs with the same parity of the number of twists are isomorphic. 
We will continue to work with colored graphs, knowing that the same reduction techniques can be applied.

Every gadget of the construction is, when considered by itself, identified by~$\LogicCdrei$. This can easily be verified by using the fact that the pairs of outer vertices have distinct colors and every two of them form a simple~$8$-cycle together with the middle vertices. It is straightforward to conclude that a graph that is indistinguishable from~$G$ or~$G'$ may only differ from them by the number of twists it has. Depending on that parity, the graph is either isomorphic to~$G$ or to~$G'$. 

\begin{theorem}
For every~$k > 2$, there is a graph~$H$ of size~$O(k)$ identified by~$\LogicCdrei$ for which the~$\LogicCK$-partition is strictly coarser than the orbit partition. Moreover, not all vertex-colored versions of~$H$ are identified by~$\LogicCK$.
\end{theorem}

\begin{proof}
Let~$G$ and~$G'$ be the graphs described in Fact~\ref{fact:CFI}. We form~$G\cup G'$, which is the graph on~$V(G) = V(G')$ with edge set~$E(G)\cup E(G')$.
This graph in particular has the parallel edges~$\{a,a'\},\{b,b'\}$ and the twisted edges~$\{a,b'\},\{b,a'\}$. Subdivide these four edges and connect the two middle vertices of the subdivided parallel edges to a new vertex~$v_p$ and the two middle vertices of the subdivided twisted edges to a new vertex~$v_t$. Let~$H$ be the resulting graph.

To argue that~$H$ is identified by~$\LogicCdrei$, we will use the fact that in the Cai, F\"urer, Immerman construction we can achieve that vertices in different gadgets have different colors. (Without this, it can still easily be proven that~$C^4$ identifies the graph.)  Since every gadget is identified by itself, the only critical edges are the
ones connecting gadgets. They can either be twisted or parallel. However, since~$G\cup G'$ has both versions present between the vertices~$a$,~$a'$,~$b$,~$b'$, for every inserted twist, there is an isomorphism to~$G\cup G'$ that resolves the twist. Thus, the graph is identified since it is isomorphic to its twisted version. Consequently, the graph~$H$ is identified by~$\LogicCdrei$ since the newly added vertices~$v_p$ and~$v_t$ are the only vertices of degree~2 and therefore distinguished by~$\LogicCdrei$ from all other vertices.

It holds that~$v_p$ and~$v_t$ are indistinguishable by~$\LogicCKminusFour$. (This can easily be seen by considering pebble games. Spoiler could win the game on~$G$ and~$G'$ by simulating the game on~$H$, placing~$4$ additional pebbles on~$a,b,a'$ and~$b'$, which would imply that~$G$ and~$G'$ can be distinguished by $\LogicCK$.)

We claim that~$v_p$ and~$v_t$ are not in the same orbit.
Assuming otherwise, there is an automorphism that maps the neighbors of~$v_p$ to the neighbors of~$v_t$ and vice versa.
This implies that there is an automorphism that maps the parallel edges to the twisted edges and vice versa. However, this implies that~$G$ and~$G'$ are isomorphic, yielding a contradiction.

Note that, since~$v_p$  and~$v_t$ are indistinguishable by~$\LogicCKminusFour$, the graph obtained by coloring~$v_p$ with a color different from the color of every other vertex is not identified by~$\LogicCKminusFour$.
\end{proof}

Even if a graph is identified and the orbits of its automorphism group are correctly determined by~$\LogicCK$, it may still be the case that this does not hold for all colored versions of the graph.

\begin{theorem}
For every~$k>2$, there is a graph~$H$ of size~$O(k)$ which is identified by~$\LogicCdrei$ such that the~$\LogicCK$-partition classes are the orbits of~$H$ but there are vertex-colored versions of~$H$ that are not identified by~$\LogicCK$ and for which the~$\LogicCK$-partition classes are not the orbits of~$H$. 
\end{theorem}

\begin{proof}
We slightly modify the construction outlined in the previous proof.
We first subdivide the parallel edges~$\{a,a'\}$ and~$\{b,b'\}$ to obtain undirected paths~$(a = a_1, a_2, a_3=a')$ and~$(b = b_1, b_2, b_3 = b')$, respectively, and also insert all twisted edges~$\{a_i, b_{i+1}\}$ and~$\{b_i,a_{i+1}\}$. 
As before, we subdivide all new edges and add for~$i \in \{1,2\}$ vertices~$v_{i,p}$ adjacent to the midpoint of~$\{a_i, a_{i+1}\}$ and to the midpoint of~$\{b_i, b_{i+1}\}$. We also add vertices~$v_{i,t}$ adjacent to the midpoint of~$\{a_i, b_{i+1}\}$ and to the midpoint of~$\{b_i, a_{i+1}\}$. 
We obtain a graph~$H$ that is, by the same reasoning as above, identified by~$\LogicCdrei$. 
Furthermore, the logic~$\LogicCdrei$ determines the orbits of~$H$ since for~$i \in \{1,2\}$ the vertex~$v_{i,t}$ can be mapped to~$v_{i,p}$ by an automorphism.
Coloring~$v_{1,p}$ yields a graph whose orbits are not the~$\LogicCK$-classes since~$v_{2,p}$ and~$v_{2,t}$ are not in the same orbit but not distinguished by~$\LogicCK$, i.e., they are in the same~$\LogicCK$-partition class.
\end{proof}

By repeating the subdivision step, the construction described in the proof can easily be generalized to show that there exist vertex-colored graphs of size~$O(k)$ in which orbits are correctly determined by~$\LogicCdrei$ even if~$k$ vertices are individualized, but in which there are~$k+1$ vertices such that individualization of all of them produces a graph whose~$\LogicCK$-partition classes are not the orbits of the automorphism group.

\section*{Acknowledgments}
We thank Konstantinos Stavropoulos and Martin Grohe for providing us with inspiring, helpful comments that guided us towards the inversion problem.

\bibliographystyle{abbrv}
\bibliography{main}

\end{document}